\documentclass[12pt]{article}
\usepackage[margin=1in]{geometry}

\usepackage[backend=biber, style=authoryear, uniquename=false, uniquelist=false, sorting=nyt, maxbibnames=99, maxcitenames=99]{biblatex}
\newcommand{\paragraphN}[1]{\paragraph{#1.}}
\usepackage{amsmath,amssymb,amsthm,mathtools}
\usepackage{mathbbol}
\usepackage{enumitem}
\usepackage{xcolor}
\usepackage{pgfplots}
\usepgfplotslibrary{groupplots}
\usepackage{url}
\usepackage[hidelinks]{hyperref}

\usepackage[capitalise]{cleveref}
\crefname{definition}{Definition}{Definitions}
\crefname{lemma}{Lemma}{Lemmas}
\crefname{proposition}{Proposition}{Propositions}
\crefname{assumption}{Assumption}{Assumptions}
\crefname{appendix}{Appendix}{Appendices}

\theoremstyle{plain}
\newtheorem{lemma}{Lemma}
\newtheorem{proposition}{Proposition}
\newtheorem{assumption}{Assumption}
\newtheorem{definition}{Definition}

\title{Platform Building with Fake Users:\\On Airdrop Farmers and Activity Procurement}
\author{
  Hanna Halaburda\thanks{
    New York University, Leonard N. Stern School of Business, Department of Technology, Operations, and Statistics, email: 
    \href{mailto:hhalaburda@gmail.com}{hhalaburda@gmail.com}}
  \and
  Benjamin Livshits\thanks{
    Imperial College London, Department of Computing, 
    email: \href{mailto:b.livshits@imperial.ac.uk}{b.livshits@imperial.ac.uk} 
    }
  \and
  Aviv Yaish\thanks{
    Princeton University, Department of Computer Science, email: 
    \href{mailto:a@yai.sh}{a@yai.sh}}
}
\date{}

\begin{document}
\maketitle
\begin{abstract}
Digital platforms reward observable claims to build their user base.
Because claims are imperfect proxies for people, these programs attract bonus hunters and airdrop farmers, who are usually treated as leakage and screened out.
We show that screening can be too aggressive when eligibility requires a separable platform-building input.
In a monopolist model, fixed per-claim rewards cannot implement finite farmer activity: rewards below farmer cost induce no entry, while rewards above cost induce unbounded entry.
A shared pool can succeed because it makes claims rivalrous.
Entry dilutes the pool until the per-claim reward equals claim cost, converting the program into identity-blind procurement.
In the baseline model, all farmer-created claims have the same lower cost, while all organic claims have the same higher cost; farmers therefore supply the activity.
Organic users receive no rewards but benefit from higher quality, so perfect screening can reduce platform profit and organic-user payoffs.
In case when primary organic claims are cheaper but additional organic identities cost more than farmer claims, a small pool attracts only organic claims, whereas a larger pool recruits farmers to expand activity.
The mechanism applies to separable inputs (e.g., liquidity, inventory, and data) rather than rewards for regular use or mechanical claim creation.
Thus, weak identity changes both leakage and how platforms procure growth.
\end{abstract}

% \noindent\textbf{Keywords:} Platforms, blockchains, airdrops, weak identity, subsidy design\\
% {\bf JEL:} D42, D47, D82, L86
% \clearpage

\section{Introduction}
\label{sec:introduction}
Digital platforms often use rewards to build markets.
They pay referral or sign-up bonuses, compensate contributors, reward listings or reviews, and distribute tokens or other benefits to accounts that satisfy eligibility criteria.
These rewards are intended to attract valuable participants or induce valuable activity.
The operational problem is that platforms often observe activity records more easily than they observe distinct participants.
A platform can see that an account was created, a referral was completed, a listing was added, a review was posted, a task was performed, or a wallet interacted with a protocol.
It may not know whether it belongs to a distinct person, a valuable long-run user, or one identity in a larger reward ``farming'' operation.
This opens a margin for exploitation \parencite{smiley2024priscila,us2021nineteen}.
In one reported case, a participant earned \$194{,}800 through DoorDash's referral system using 487 accounts \parencite{cimpanu2021brazilian}.
PayPal disclosed that it had identified 4.5 million accounts that it believed were illegitimately created in connection with increased use of incentive campaigns \parencite{paypal2022q4,finextra2022paypal}.

This weak-identity friction is especially visible in blockchain airdrops.
Airdrops are commonly described as platform rewards to users or community members.
In practice, eligibility is often attached to pseudonymous wallets based on observable activity histories.
A blockchain platform may distribute tokens to wallets that traded, bridged assets, supplied liquidity, voted, referred others, or interacted with a smart contract.
While activity is publicly observable, the person behind the wallet is not.
One person may control many wallets, and a wallet that satisfies an eligibility rule need not represent a distinct long-term user.
Airdrops therefore provide a convenient setting for studying a broader platform problem: rewards are frequently paid to observable claims rather than to perfectly verified individuals.
The dollar stakes can be large.
Uniswap, a decentralized exchange, launched its UNI token in 2020 by making 400 UNI claimable by each address that had previously interacted with Uniswap v1 or v2; this was worth roughly \$1{,}200 per eligible address at the time \parencite{uniswap2020introducing,tunguz2022uniswap}.
Arbitrum, an Ethereum scaling network, allocated 1.162 billion ARB tokens (11.62 percent of its initial supply) to user wallets; at a launch-day price around \$1.38, this implied roughly \$1.6 billion in user-airdrop value \parencite{arbitrum2023arb,coindesk2023arbitrum}.
Jito, a Solana liquid-staking protocol, allocated 100 million JTO tokens to a retrospective airdrop, valued at about \$165 million \parencite{jito2023jto}.
These examples also illustrate two common reward architectures.
Some programs specify an amount per eligible claim record.
Others commit an aggregate allocation and then use points, scores, tiers, or rankings to determine how this pool is divided among individual claims.

Once rewards are attached to claim records, a new supply-side effect appears.
Agents can specialize in producing eligible claims.
In blockchain markets, these agents are often called airdrop farmers.
In other settings, they appear as referral fraudsters, bonus hunters, duplicated accounts, fake reviewers, or task-completion farms.
Evidence from airdrops is consistent with substantial reward-driven participation.
\textcite{fan2023altruistic} report that 41 percent of ParaSwap recipients sold their rewards immediately and that 61 percent stopped using the platform.
Roughly half of Uniswap airdrop recipients stopped using the platform within a month, with the share rising substantially over the following year \parencite{jhackworth2022uniswap}.
Analysis of the Arbitrum airdrop has estimated that nearly 50 percent of recipient addresses were same-person addresses, with almost 22 percent belonging to clusters of 20 or more related addresses \parencite{0x302a2023advanced}.
The usual interpretation is that these claimants are undesired leakage.
They collect transfers without being the users or contributors the platform meant to subsidize.
The natural prescription is therefore to detect them, screen them out, and reserve rewards for verified participants whenever possible.

We show that the leakage interpretation is incomplete.
The key distinction is between rewards for regular use and for separable platform-building activity.
For example, if a reward is paid just for signing up or logging in, then reward-seeking claimants are indeed profit-reducing leakage.
They receive transfers without creating independent platform value, or they could even create congestion.
But if the rewarded action creates a separable input into platform quality, the same claimants may play a different economic role.
A claimant who is not a valuable long-run user may still supply liquidity, capital, inventory, content, reviews, data, integrations, documentation, or other complements that make the platform more valuable to organic users.

Our central contribution is to identify the procurement channel created by weak identity.
When eligibility is based on observable claims rather than verified persons, a platform may not be subsidizing the intended users directly.
Instead, it may be buying the rewarded activity from whoever can produce eligible claims at lowest cost.
The platform is then not paying fake users because they are useful consumers.
It is buying a platform-building input from low-cost claimants whose identities are imperfectly verified.
Such procurement shifts the demand curve by raising the platform's value to organic users.
A reward for regular use instead acts like a price discount: it moves users along the base-quality demand curve rather than changing the quality of the platform.

The mechanism follows from the interaction between weak identity and rewards.
A non-rival per-claim reward is fragile when farmers can create additional eligible identities.
If the reward is below farmers' claim cost, farmers do not enter.
If it is above it, each claim is profitable and farmer claim creation will be unbounded.
If it is exactly at cost, the activity level is not pinned down.
A shared reward pool changes the logic because each additional claim dilutes the payment received by all claimants.
Competitive farmer entry then drives the realized reward down to farmers' claim cost.
In the baseline model, all farmer-created claims have the same lower cost, while all organic claims have the same higher cost, so the realized reward is too low for organic users to claim and the platform procures the rewarded activity only from farmers.
Appendix~\ref{sec:primary_organic_advantage} considers the alternative ordering in which primary organic claims are cheaper than farmer claims but additional organic identities are more costly.
A small pool then attracts only primary organic claims, while a larger pool recruits farmers to expand activity beyond that inexpensive organic margin.

This mechanism leads to our main welfare result.
When the rewarded task is a separable platform-building input and farmers are the low-cost suppliers, perfectly screening farmers out can backfire: while it eliminates reward leakage, it also removes the low-cost procurement channel.
If that platform-building activity must instead be supplied by verified organic agents at higher cost, the platform procures less of it.
Platform profit can fall.
Organic users can also be worse off, even though farmers would have received the rewards, because screening deprives them of the platform-quality improvement created by the additional activity.
Thus, the agents who receive the reward and the agents who benefit from the rewarded activity need not be the same.
Farmers receive the payments, while organic users may benefit from the induced improvement in platform quality.

The result is conditional and does not imply that all farming is valuable.
What matters is whether the claim is tied to an action that creates monetizable platform value for organic users.
When the rewarded action creates no such value, reward-seeking claimants are leakage and screening remains a natural response.
The paper's claim is narrower than a general defense of reward farming: weak identity can make farmer exclusion inefficient when rewards are attached to separable platform-building activity.

We formalize this argument in a platform model that isolates the weak-identity friction.
A monopolist platform chooses an access fee and an airdrop tied to separable platform-building tasks, but it observes eligible task completions rather than verified persons.
Organic users value platform access and may join at a fee.
Farmers assign no value to platform access and participate only when rewards cover their claim cost.
The baseline abstracts from direct network effects, so any benefit from farming must come through the procured task rather than through farmer adoption.

In formal analysis we show that weak identity transforms a pooled airdrop into activity procurement: the platform buys the rewarded task from the lowest-cost claimants rather than subsidizing organic users directly.
Moreover, relative to a verified benchmark in which farmers are screened out and the platform-building activity must be procured from organic agents, weak identity can raise platform profit and every organic user's payoff.
We then study the scope of the mechanism.
If farmer-supplied activity is lower value than organic-supplied activity, the comparison becomes a cost-quality tradeoff.
If the airdrop is issued in a multi-period environment, residual persistent value may raise the return to airdrop-funded procurement because the platform pays the claim cost once but may monetize the resulting platform value over time.
Alternatively, with no residual value, it may be optimal for the platform to issue airdrops every period.

The implication is that rather than tolerating all farming or trying to fully abolish it, platforms can harness it to efficiently procure valuable activity.
For rewards tied to access or regular use, screening remains central.
A reward for regular use is essentially a price discount unless the rewarded action benefits the platform.
For rewards tied to separable platform-building inputs, however, perfect screening may be too aggressive.
The relevant questions are thus whether (i) the task creates value, (ii) the reward structure controls unbounded claim creation, and (iii) excluding farmers would raise platform building cost.

The rest of the paper proceeds as follows.
\cref{sec:related} situates the paper in the literatures on platform subsidies, token incentives, weak identity, and unauthorized access.
\cref{sec:background} describes the institutional features of airdrops and related reward programs that motivate the model.
\cref{sec:model} presents the baseline model.
\Cref{sec:results} presents the main results, showing how weak identity turns airdrops into activity procurement.
\Cref{sec:robustness_extensions} clarifies the scope of the model and presents extensions.
Specifically, \cref{sec:fee_discount} explains why rewards for regular use differ from rewards for separable platform-building activity; \cref{sec:lower_value_farmer_activity} studies how the comparison changes when farmer-supplied activity is lower value than verified organic-supplied activity; and \cref{sec:dynamic} studies persistence and durable platform value.
The main paper concludes with \cref{sec:conclusion}.
\Cref{sec:heterogeneous,sec:competition} extend the analysis to heterogeneous claim costs and competition.

\section{Related literature}
\label{sec:related}

To the best of our knowledge, this is the first formal economic analysis of airdrops under weak identity that studies when reward-seeking claimants can be profitably harnessed.
This section situates our work in four literatures: platform subsidies and market building, token incentives and airdrops, weak identity and Sybil behavior, and digital piracy.
The closest platform-subsidy papers study how prices and transfers build participation when adoption creates value.
The closest weak-identity papers study how to detect, deter, or design around false-name manipulation.
We connect these themes by studying a setting in which reward eligibility is attached to weakly verified claims.

\paragraphN{Platform subsidies and market building}

Our analysis is closely related to the literature on platforms, where subsidies and market building have been studied extensively.
A central insight of this literature is that platforms may use prices, transfers, or below-cost introductory offers to build a market when participation creates value.
\textcite{cabral1999monopoly} study introductory pricing by a monopolist facing network externalities.
\textcite{caillaud2003chicken} and \textcite{armstrong2006competition} analyze how two-sided platforms allocate prices and subsidies across sides when participation by one group benefits another.
\textcite{bakos2020platform} show that this subsidy logic can change when agents multihome on both sides.
These papers build on the broader network-goods literature on adoption, installed base, compatibility, and market formation, including \textcite{rohlfs1974theory}, \textcite{katz1985network}, \textcite{economides1996economics}, and \textcite{fudenberg2000pricing}.
We take the market-building logic in this literature as our starting point: the platform pays agents today in order to create monetizable platform value.

Our departure is the implementation environment.
Canonical models usually take the subsidized unit to be a well-defined participant: a user, adopter, buyer, seller, or side of the market.
This is a natural abstraction for studying how much to subsidize, which side to subsidize, or how subsidies affect adoption over time.
Some work also recognizes that subsidies can attract opportunistic participation.
For example, \textcite{caillaud2003chicken} show that users may register with an entrant to collect a subsidy while maintaining incumbent access to preserve externality benefits.
In their model, opportunistic subsidy collection occurs within a fixed population of users.
In our setting, weak identity makes the mass and composition of eligible claims endogenous: agents can create additional accounts or claims when doing so is profitable.
Moreover, the reward is not limited to lowering access price for the organic users the platform wants to serve.
It may be attached to an eligible activity that can be supplied by claimants whose primary motive is the reward.
Thus, our mechanism is related to platform-subsidy logic, but the object being implemented is different: the platform is not only choosing a subsidy level, but whether the reward is implemented as a subsidy to organic users or as procurement of a separable platform-building input from low-cost claimants.

\paragraphN{Airdrops and token incentives}

A growing literature studies how blockchain tokens expand the set of instruments available to platforms.
Token distributions can finance growth using claims on future platform value, align incentives through ownership and governance, and facilitate coordination \parencite{bakos2018role,bakos2022overcoming,rogoff2022redeemable,chod2021theory,chod2022platform,malinova2023tokenomics,li2025digital}.
Related work studies how token rewards can be used to attract users from competing platforms and coordinate migration between ecosystems \parencite{hatfield2023simple,gan2023decentralized}.
Airdrops, i.e., token distributions to users, have emerged as a widely used implementation of these ideas.
\textcite{georganas2025airdrop} study the allocation problem in such distributions: how assigning tokens to potential contributors affects their subsequent contribution incentives.

Existing work documents the objectives, design choices, and empirical effects of airdrops \parencite{allen2023crypto,allen2023why,lommers2023designing,makridis2023rise,li2024operation}.
A recurring finding is that airdrops attract substantial reward-claiming activity despite platforms' efforts to screen or exclude farmers \parencite{fan2023altruistic,messias2023airdrops}.
\textcite{luo2025resilient} combine empirical measurement of hunter profits with a game-theoretic analysis of mechanisms for detecting and discouraging farmers.
We study a different question: Rather than treating airdrops as a token-allocation or hunter-detection problem, we ask when a platform may want reward-seeking claimants to participate because they are low-cost suppliers of platform-building tasks.

\paragraphN{Weak identity and Sybil behavior}

A large literature studies how digital platforms can limit manipulation by agents operating multiple identities.
The concept of a Sybil attack originates in the study of decentralized systems, where a single agent can create multiple identities to gain disproportionate influence \parencite{douceur2002sybil}.
Subsequent work develops proof-of-personhood mechanisms and account-clustering techniques intended to detect or deter false identities \parencite{siddarth2020who,victor2020address}.
A related mechanism-design literature studies how incentive schemes can be made robust to false identities and Sybil attacks \parencite{emek2011mechanisms,drucker2012simpler,yaish2026inequality}.
Recent work applies these ideas directly to airdrops, proposing methods for identifying Sybil accounts and reward claimants using on-chain data \parencite{liu2022fighting,zhou2024artemis}.

Our perspective is different.
Most of this literature views Sybil accounts and reward farmers as a source of inefficiency to be detected, excluded or deterred.
By contrast, we study environments in which some claimants remain indistinguishable from legitimate participants, even after platforms expend resources to identify them.
The question is therefore not how to eliminate farmers, but what their presence implies for platform economics when the platform can condition rewards on a separable platform-building task.
In that environment, a weakly verified claim is not only a potential source of leakage; it can also be a source of low-cost input supply.

\paragraphN{Digital piracy and unauthorized access}

Our analysis is also related to the economics of digital piracy and unauthorized access.
The conventional concern in this literature is that pirating weakens appropriability: users who might otherwise purchase the product can obtain access without paying, reducing revenue and the return to investment in quality \parencite{besen1989private,hui2003piracy,belleflamme2012digital}.
At the same time, the literature has recognized that unauthorized access need not always reduce firm profit.
When information goods exhibit demand-side network externalities, unauthorized users can expand the effective user base and increase the value of the product to paying users \parencite{conner1991software,takeyama1994welfare}.
More generally, work on sharing, copying, and renting information goods shows that non-purchase access can reshape demand or segment users in ways that may increase profits under some conditions \parencite{bakos1999shared,varian2000buying}.
\textcite{jain2008digital} studies a piracy-specific version of this logic in a competitive market, where unauthorized copying can soften price competition by allowing price-sensitive consumers to obtain unauthorized copies.

The connection between this literature and our work is that agents outside regular paying demand can affect the value generated from paying users, though this is done via a different channel.
In piracy models, non-paying actors get access to the focal product, so the main question is how unauthorized consumption affects the market (e.g., prices and quality).
In our model, farmers need not consume the platform at all.
Their payoff-relevant action is completion of a reward-eligible platform-building task through weakly verified claims.
Thus, the relevant friction is procurement of a platform input under weak identity rather than unauthorized access to the product.

\section{Institutional background: Airdrops, reward farming, and platform-building tasks}
\label{sec:background}

This section records the institutional facts that map the airdrop setting into the model.
We emphasize five features: rewards are attached to observable claims; claim identities imperfectly proxy for individuals; some claimants specialize in reward extraction; eligibility tasks differ in platform-building value; and airdrop-induced activity may be temporary or persistent.

\subsection{Airdrops reward observable claims}

Airdrops are reward programs in which a platform distributes tokens or other benefits to participants who satisfy eligibility criteria.
The criteria vary widely across programs.
Case studies document airdrops based on past platform use, transaction histories, liquidity provision, governance participation, referrals, or completion of platform-specific tasks \parencite{allen2023crypto,allen2023why,lommers2023designing}.
Recent legal commentary makes the same institutional point from a different angle: contemporary airdrops are often not pure giveaways, but campaigns in which participants expend effort, time, or capital in anticipation of token rewards while projects use the induced activity to build adoption and market traction \parencite{sie2025regulating}.

The important institutional feature is that eligibility is implemented through observable claim records.
The platform observes that an account, address, or other record satisfies a rule.
It need not observe whether two records are controlled by the same entity or by two distinct ones.

Reward schedules also vary.
Some distributions specify explicit amounts per eligible address or criterion, as in Uniswap and Optimism.
Other programs first fix an aggregate allocation and then map activity scores, points, tiers, or rankings into individual claims, as in Arbitrum and Jito \parencite{uniswap2020introducing,optimism2022airdrop,arbitrum2023arb,jito2023jto}.
Thus, the amount received by one claimant may either be independent of other claimants or diluted by the number of other claims.

\paragraphN{Modeling implication}
We model an airdrop as a reward attached to an eligible task completion.
The reward schedule includes both a per-claim component and a shared pool component.
This specification is not meant to reproduce every institutional formula exactly; it captures the distinction between non-rival per-claim rewards and budgeted reward allocations whose realized value depends on the number of eligible claims.

\subsection{Claim identities are imperfect proxies for individuals}

Account identities are not the same as distinct individuals.
In blockchains, platforms usually observe addresses or accounts, and one person may control many addresses.
\textcite{bensasson2024airdrop}, discussing the Starknet airdrop, describes the issue directly: blockchains observe ``account addresses, not people.''
The same issue appears outside blockchain settings whenever rewards are attached to accounts, referrals, profiles, or other claim records that can be duplicated or strategically produced.

Platforms use several tools to manage this identity uncertainty, e.g., retroactive eligibility rules, behavioral screens, and ex post disqualification.
Technical literature studies methods for detecting or clustering Sybil accounts in airdrops and related settings \parencite{liu2022fighting,zhou2024artemis}.
These tools can make duplicated claiming more difficult, but existing evidence shows that screening is incomplete.
\textcite{fan2023altruistic}, for example, study the ParaSwap airdrop and report that some Sybil accounts were not detected.

Retroactive designs do not remove the identity friction.
They can make airdrops harder to target strategically before announcement, but prospective claimants may still anticipate future distributions and build account histories intended to qualify.
Industry reporting describes organized airdrop farming in precisely these terms \parencite{copeland2023we}, and airdrop design studies emphasize that eligibility rules and signaling can affect who prepares for future distributions \parencite{lommers2023designing}.

\paragraphN{Modeling implication}
The model assumes weak identity: the platform observes eligible task completions but cannot perfectly distinguish organic users from farmer-created identities.
The verified benchmark studies the counterfactual case in which this friction is removed.

\subsection{Some claimants specialize in reward extraction}

Airdrops attract many participants who are primarily motivated by rewards, rather than platform use.
Industry reporting describes airdrop farming as the organized creation or management of eligible accounts to increase expected claims \parencite{copeland2023we}.
Post-airdrop behavior is consistent with this reward-extraction motive: studies document rapid selling and exit after several airdrops \parencite{jhackworth2022uniswap,fan2023altruistic,messias2023airdrops}.
This evidence does not imply that every exiting recipient assigns no value to platform participation; a genuine user may sell for portfolio reasons.
It does support, however, a modeling distinction between organic users who value platform access, and reward-motivated claimants whose participation hinges on whether rewards cover claiming costs.

A related point concerns the marginal cost of producing eligible identities.
Organic users wishing to perform a rewarded action through a second or third account face the cost of establishing a credible account history, satisfying behavioral screens, and managing separate identities.
Specialized farmers can reuse wallet infrastructure and automated account-management and eligibility routines across many identities.
Industry reporting describes farmers as multi-wallet actors who use automated or randomized on-chain actions, and \textcite{luo2025resilient} document coordinated funding, bridging, and consolidation patterns among reported airdrop hunters.
This specialization does not imply that farmers are necessarily cheaper than organic users at a primary account.
An organic user may perform the qualifying task cheaply through an existing account.
The potential farmer advantage concerns repeated claim production: after inexpensive primary organic claims are exhausted, farmers may be cheaper than inducing organic users to establish additional qualifying identities.

\paragraphN{Modeling implication}
The baseline assumes identical claim costs within each group: every farmer-created claim costs \(c_F\), while every organic claim costs \(c_O>c_F\).
This isolates the limiting case in which farmer-created claims are cheaper at every relevant margin.
\Cref{sec:primary_organic_advantage}  imposes the alternative ordering in which primary organic claims are cheaper than farmer claims but additional organic identities are more costly.
It shows that a small pool can endogenously screen farmers, while a larger pool recruits them at the expansion margin.
\Cref{sec:heterogeneous} otherwise allows general heterogeneous claim costs.
Moreover, farmers receive no standalone utility from platform access and enter only when the realized reward covers their claim cost.

\subsection{Eligibility tasks differ in the platform value they create}

Airdrop eligibility criteria are not economically equivalent.
Some rewards are attached to regular use, such as signing up, logging in, making a transaction, or maintaining an account.
Other rewards are attached to actions that can create a separable platform input: supplying liquidity, contributing capital, listing inventory, creating content, writing reviews, referrals, NFT minting, labeling data, building integrations, or producing other complements.
Airdrop case studies document this range of eligibility criteria \parencite{allen2023crypto,allen2023why,sie2025regulating}.
In decentralized-exchange settings, liquidity provision is a natural example of an action that may improve platform quality: additional liquidity can improve execution conditions for traders, while regular platform use is trading itself.
In marketplace or software-platform settings, listings, reviews, content, data, and integrations can similarly make the platform more useful to other users.

The first distinction is whether the action improves the platform for other users.
For example, liquidity provision, listing, review, data label, or integration may be useful even if the provider is otherwise not a platform user.
By contrast, an eligible action may fail to build platform value if it creates claim records without improving quality.
It may even reduce value if it generates activity that imposes costs on other users, such as congestion or noise trading.

The second distinction is whether the output of the rewarded action is source-neutral.
For some tasks, the platform value of the input may be similar regardless of whether it is supplied by an organic user or by a reward farmer.
Liquidity, capital supply, inventory, verifiable data, and some integrations may have this feature when the input is useful on its own.
For other tasks, the source of the activity may matter.
Reviews, content, governance participation, or some data labels may be more valuable when generated by organic users than by reward farmers.
Thus, even among separable platform-building tasks, farmer-supplied activity may be lower value than organic-user-supplied activity.

Whether a given eligible action creates independent value is ultimately an empirical question that depends on the specific task and platform.
Existing work is consistent with heterogeneous effects.
\textcite{makridis2023rise} find positive platform-level effects associated with airdrops and governance tokens on decentralized exchanges, and \textcite{guo2023spillover} find spillover effects on engagement across tokenization platforms.
Studies documenting rapid exit and selling \parencite{fan2023altruistic,messias2023airdrops} are consistent with many airdrops rewarding activity that does not create durable user participation.
These findings motivate separating rewarded actions that create platform inputs from rewards for regular use or mechanical activity.

\paragraphN{Modeling implication}
The baseline model focuses on a separable platform-building task $A$: the value it creates for organic users does not require the claimant to remain an active platform participant.
The baseline adopts the source-neutral benchmark, in which a unit of the task raises organic users' gross value from $v(0)$ to $v(A)$ regardless of whether the activity is supplied by farmers or organic users.
This isolates the procurement-cost channel.
\cref{sec:lower_value_farmer_activity} relaxes source-neutrality by allowing farmer-supplied activity to create lower platform value than verified organic-supplied activity.
\cref{sec:fee_discount} studies the boundary case in which the reward is attached to regular use rather than to a separable platform-building input.

\subsection{Airdrop effects can be temporary or persistent}

A final institutional issue is persistence.
Many reward recipients sell quickly or stop using the platform after the airdrop.
This makes it risky to model airdrop recipients as durable users.
Evidence from ParaSwap, Uniswap, and other airdrops points to substantial post-airdrop exit and selling \parencite{jhackworth2022uniswap,fan2023altruistic,messias2023airdrops}.

At the same time, recipient exit does not imply that every rewarded action loses all value immediately.
Some actions may leave behind assets that remain useful after the claimant exits: content, listings, reviews, data, reputation, integrations, documentation, or other complements.
Other actions, such as temporary liquidity or temporary capital provision, may be less persistent but can still create value while present.
Empirical work is mixed across outcomes: some studies document short-lived recipient activity, while others find effects on platform-level measures such as engagement, market capitalization, or total value locked that may be more persistent than short-run account activity \parencite{guo2023spillover,makridis2023rise,messias2023airdrops}.

Thus, the institutional evidence supports two separate possibilities.
Some airdrops may procure only temporary activity.
Others may procure activity that leaves a durable platform stock.
Which case applies depends on the task being rewarded.

\paragraphN{Modeling implication}
The baseline model values the contemporaneous effect of rewarded activity through $v(A)$.
The dynamic extension in \cref{sec:dynamic} allows rewarded activity to leave a residual stock of platform value after the airdrop has been paid.

\section{Model}
\label{sec:model}

We study a monopolistic platform that uses an airdrop to build adoption when identities are imperfectly verifiable.
The model aims to isolate the mechanism in the starkest possible way.
We abstract from direct adoption network effects and focus on a separable platform-building input.\footnote{
  Adding the term $\gamma\, n$ to the user's utility function, where $\gamma$ measures strength of direct adoption network effects, adds $\gamma$ to some formulas, but does not change the results.
  Full model with network effects is available upon request.}
This lets us ask whether farmer participation can be valuable even when farmers have no consumption value, do not pay access fees, and are useful only through the activity they supply.
The answer is \emph{yes}.
Weak identity does not merely create reward leakage: when the reward is attached to a platform-building activity, competitive farming can turn the airdrop into an identity-blind mechanism for buying that activity from the lowest-cost claimants.

\paragraphN{Platform-building activity}
Let $A\in\mathbb R_+$ denote the total incremental amount of a rewarded platform-building activity.
Organic users value the platform according to a gross value function $v\! : \mathbb R_+\to\mathbb R_+$, where $v(A)$ is increasing and concave, with $v(0)>0$.
The value $v(0)$ is the platform's standalone value absent additional rewarded activity.
The variable $A$ is incremental activity procured by the reward program; existing infrastructure, unrewarded organic activity, and other baseline sources of quality are absorbed into $v(0)$.

The key condition is that the rewarded task is separable from regular platform use.
The claimant who performs the task need not be an active organic user for the task to raise platform value.
Examples include liquidity provision, capital supply, inventory listing, content creation, reviews, data labeling, documentation, and integration building.
The model does not apply to rewards for regular platform use that create no independent platform value.
A reward for signing up, logging in, trading solely to generate volume, or making a mechanical transaction is closer to a price discount unless the action also creates a separable improvement in platform quality.

In the baseline model, each active identity can complete at most one unit of the rewarded task.
This keeps the identity margin explicit, which is the margin that matters under weak identity.
\Cref{sec:heterogeneous} relaxes this homogeneous one-identity technology.

\paragraphN{Reward policy and platform profit}
A platform policy consists of an access fee $p$ and a reward schedule $(r,\Omega)$.
If total rewarded activity is positive, each rewarded identity receives
\[
    \tau(A)=r+\frac{\Omega}{A}, \quad \text{ for } A>0,
\]
where $r\ge0$ is a per-task reward and $\Omega\ge0$ is a shared reward pool.
If $A=0$, no reward is paid; for notational convenience set $\tau(0)=0$.
The total transfer cost is\footnote{
  The variables $r$, $\Omega$, and $\tau(A)$ are measured in claimant-value units.
  The baseline normalizes the platform's cost of delivering one unit of claimant value to one.
  This is a reduced-form transfer-cost benchmark, not an assumption that token issuance or dilution is literally linear in token counts.
  For token airdrops, $T(A)$ can be interpreted as the market value, opportunity cost, or treasury cost of the reward package delivered to claimants.
  If the platform's cost of delivering aggregate reward value $T$ is nonlinear, write it as $\mathcal C(T)$, with $\mathcal C(0)=0$ and $\mathcal C$ increasing.
  Replacing the linear delivery-cost benchmark with $\mathcal C(T)$ leaves the implementation results below unchanged because farmer entry is governed by claimant value $\tau(A)$, not by the platform's accounting cost of delivering that value.
  A nonlinear delivery cost changes the platform's optimal scale of activity, but not the equilibrium logic by which a pooled reward is competed down to farmers' claim cost.
  These delivery-cost effects are distinct from endogenous token-price effects and from the consequences of who receives and retains the tokens.
  When the pool is specified as a quantity of tokens, its claimant value may depend on issuance and anticipated post-airdrop sales.
  The discussion at the end of \cref{sec:lower_value_farmer_activity} considers these additional token-specific channels.} $$T(A)=rA+\Omega\cdot\mathbb{1}\{A>0\}.$$
This reward schedule allows for pure fixed rewards $(r>0,\Omega=0)$, pure pooled rewards $(r=0,\Omega>0)$, and mixed rewards $(r>0,\Omega>0)$.

Platform profit is $\Pi=pn-T(A)$, where $n$ denotes the mass of organic adopters.

\paragraphN{Organic users}
The platform serves a unit mass of potential organic users indexed by $\theta\in[0,1]$, where lower $\theta$ indicates a better match with the platform.
The platform charges an access fee $p\ge0$ for regular platform use.
If organic user $\theta$ adopts but does not perform the rewarded task, her payoff is $$u^0_\theta(A,p)=v(A)-\theta-p.$$
If she adopts and performs one rewarded task, her payoff is $$u^1_\theta(A,p)=v(A)-\theta-p+\tau-c_O,$$ where $c_O>0$ is the cost of supplying one eligible organic claim.
The outside option is normalized to zero.

The access fee $p$ is charged for regular platform use.
Any costs required to complete the rewarded task (e.g., due to fees or effort) are included in $c_O$ and, below, in $c_F$.
\footnote{
  The baseline model does not distinguish a user's primary claim from claims produced through additional identities; \cref{sec:primary_organic_advantage} makes that distinction.}

\paragraphN{Farmers and weak identity}
In addition to organic users, there is a competitive fringe of atomistic farmers with unrestricted total mass.
The platform cannot perfectly distinguish a farmer-created identity from an organic one.
Each farmed identity can complete one unit of the rewarded task at cost $c_F$, where $0<c_F<c_O$.
Farmers obtain no standalone utility from regular platform use.
Because the rewarded task is separable from regular access, a farmer who completes it bears the task cost but does not pay the platform access fee $p$.
A farmer-created identity's payoff is therefore $\tau-c_F$.
Let $a_O$ denote organic-supplied rewarded activity and $a_F$ farmer-supplied rewarded activity.
The total rewarded activity is $A=a_O+a_F$.

The inequality $c_F<c_O$ is the reduced-form representation of lower-tail cost dominance: specialized reward claimants are the lowest-cost suppliers of marginal eligible identities.
The heterogeneous-cost extension in \cref{sec:heterogeneous} replaces the homogeneous-cost assumption with a distributional condition.

\paragraphN{Timing}
The interaction is sequential.
First, the platform announces $(p,r,\Omega)$.
Second, organic users decide whether to adopt and whether to complete the rewarded task, while farmers decide how many identities to create and send into the reward program.
Finally, rewarded activity, transfers, and payoffs are realized.

\paragraphN{Equilibrium}
We solve the game by backward induction.
Some platform policies $(p,r,\Omega)$ may induce unbounded farmer entry in the subsequent subgame.
In our analysis, bounded entry plays an important role.
The following definition characterizes finite continuation outcomes in the user and farmer entry subgame, for a fixed platform policy.
Let $\mathcal N$ denote the set of organic adopters and $\mathcal C$ the set of organic users who both adopt and complete the rewarded task.
To simplify notation, ties are broken in favor of adoption, and conditional on adoption, in favor of claiming.
This convention affects only cutoff types and does not change the main results.

\begin{definition}[Bounded continuation equilibrium]
\label{def:equilibrium}
Fix a policy $(p,r,\Omega)$.
A bounded continuation equilibrium is a tuple $(n,a_O,a_F,A)$ where $a_F$ is finite and there exist sets $\mathcal C\subseteq\mathcal N\subseteq[0,1]$ such that $n=\mu(\mathcal N)$ and $a_O=\mu(\mathcal C)$, with $\mu$ denoting Lebesgue measure, and such that:
\begin{enumerate}[label=(\roman*), leftmargin=2em]
\item organic users optimize, i.e., for almost every $\theta\in[0,1]$:
\begin{itemize}[leftmargin=1.5em]
\item[] if \(\theta\in\mathcal C\), then
        \[
            u^1_\theta(A,p)\ge \max\{u^0_\theta(A,p),0\};
        \]
\item[] if \(\theta\in\mathcal N\setminus\mathcal C\), then
        \[
            u^0_\theta(A,p)\ge 0
            \quad\text{and}\quad
            u^0_\theta(A,p)>u^1_\theta(A,p);
        \]
\item[] if \(\theta\notin\mathcal N\), then
        \[
            0>\max\{u^0_\theta(A,p),u^1_\theta(A,p)\}.
        \]
\end{itemize}

\item farmers satisfy the finite free-entry condition: if $a_F>0$, then $\tau(A)=c_F$; if $a_F=0$ and $A>0$, then $\tau(A)\le c_F$; and if $a_F=0$ and $A=0$, then $\Omega=0$ and $r\le c_F$.
\item total rewarded activity satisfies $A=a_O+a_F$.
\end{enumerate}
\end{definition}

The first condition in \cref{def:equilibrium} states individual optimality for organic users.
Because the incremental payoff from completing the task is $$u^1_\theta(A,p)-u^0_\theta(A,p)=\tau(A)-c_O,$$ which does not depend on $\theta$, the continuation equilibrium will have a simple structure: if the realized reward covers the organic task cost, organic adopters complete the task; if it does not, they do not.
The adoption margin is then a threshold in $\theta$, derived from the payoff of the best available adoption option.

The second condition is the finite-outcome version of farmer free entry.
If a positive finite mass of farmers enters, entry must drive the per-identity payoff to zero, so $\tau(A)=c_F$.
If no farmers enter, no farmer-created identity can earn a positive payoff.
The separate $A=0$ clause is needed because a positive reward pool cannot be sustained with zero claimants: if $\Omega>0$, an arbitrarily small positive mass of farmers could enter and obtain a large pooled reward.
Thus, no-entry at $A=0$ requires $\Omega=0$ and $r\le c_F$.

We do not require the bounded continuation equilibrium to be unique.
However, the platform's optimization problem requires the induced amount of rewarded activity to be well-defined.
We say that a policy is an \emph{activity-determinate implementation} of activity $A$ if it induces at least one bounded continuation equilibrium and every bounded continuation equilibrium induced by the policy has total rewarded activity $A$.
Policies that induce unbounded farmer entry or leave positive activity indeterminate are not well-defined airdrop policies for the platform's profit maximization.

\section{Airdrops as activity procurement}\label{sec:results}
\subsection{Weak identity turns airdrops into activity procurement}
\label{sec:activity_procurement}

The first result shows why a non-rival per-claim reward is not a robust implementation of platform-building activity when identities are weak.
If the fixed reward is below farmer cost, farmers do not enter.
If it equals farmer cost, the amount of farmer activity is not pinned down.
If it exceeds farmer cost, entry is unbounded.

\begin{lemma}[Pure fixed rewards cannot implement positive finite activity]
\label{res:fixed}
No policy with pure fixed reward $\Omega=0$ can be an activity-determinate implementation of a positive finite rewarded activity.
Specifically:
\begin{enumerate}[label=(\roman*), leftmargin=2em]
\item If $r<c_F$, then every bounded continuation equilibrium has $A=a_O=a_F=0$; \item if $r=c_F$, every bounded continuation equilibrium has $a_O=0$, and every finite $a_F\ge0$ can be supported as a bounded continuation equilibrium, so activity is indeterminate; \item If $r>c_F$, then no bounded continuation equilibrium exists.
\end{enumerate}
\end{lemma}

\begin{proof}
In the appendix.
\end{proof}

A pooled component changes the logic because it makes the reward rivalrous.
Each additional claimant dilutes the shared pool.
In any bounded equilibrium with a finite mass of farmed identities and positive rewarded activity, entry continues only until the realized reward is driven down to farmer unit cost, $c_F$.
Because $c_F<c_O$, that reward is too low to induce organic users to complete the task.
The airdrop therefore procures the activity from the lowest-cost claimants.

\begin{proposition}[Weak identity transforms airdrops into activity procurement]
\label{res:activity_procurement}
Fix a policy $(p,r,\Omega)$ with $\Omega\ge0$ and $r<c_F$.
In any bounded continuation equilibrium,
\[
    a_O=0,\qquad a_F=A,\qquad
    A=\frac{\Omega}{c_F-r}.
\]
Organic adoption is $n=\min\{1,\max\{0,v(A)-p\}\}.$
If adoption is interior, then $n=v(A)-p$.
\end{proposition}

\begin{proof}
In the appendix.
\end{proof}

\cref{res:activity_procurement} is the main implementation reduction.
For any target $A\ge0$, the platform can implement $A$ by choosing any $r<c_F$ and setting $\Omega=(c_F-r)A$.
If $A>0$, the realized reward is $\tau(A)=c_F$, farmers supply all rewarded activity, and total transfers are $T(A)=rA+\Omega=c_FA$.
Thus the platform's reward-design problem collapses to a procurement problem: choose how much separable activity to buy at unit cost $c_F$, then choose the access fee charged to organic users.

The platform may not know the exact farmer cost when it commits to a pool.
Suppose the platform commits to a pure reward pool \(\Omega>0\) before learning a common per-claim farmer cost \(C_F\), while farmers observe \(C_F\) before entering.
Conditional on the realized cost, free entry implies $     \frac{\Omega}{A}=C_F$, and hence $   A(C_F)=\frac{\Omega}{C_F}$.
The aggregate reward value is fixed at \(\Omega\), but the amount of activity procured by the pool depends on the realized cost: a lower farmer cost produces more activity, while a higher cost produces less.
By contrast, a fixed per-claim reward induces no farmer activity when it is below \(C_F\), leaves activity indeterminate when it equals \(C_F\), and induces unbounded entry when it exceeds \(C_F\).
Thus, Proposition~\ref{res:activity_procurement}'s implementation logic continues to hold once \(C_F\) is realized.
With cost uncertainty, however, the platform cannot guarantee a particular activity level ex ante; through its choice of \(\Omega\), it can instead target an expected level of activity.
For the remainder of the analysis, we return to the baseline case in which \(c_F\) is known.

We now solve the pricing and activity-choice problem that the platform faces.
Hats denote objects that are optimized conditional on a fixed activity level $A$.

\begin{lemma}[Optimal access pricing and activity procurement]
\label{res:optimal_pricing_activity}
Fix $A\ge0$.
Conditional on $A$:
\begin{enumerate}
\item[(i)] the platform-optimal organic adoption level is
\[
    \widehat n(A)=\min\left\{\frac{v(A)}{2},1\right\};
\]
\item[(ii)] the corresponding fee is
\[
    \widehat p(A)=v(A)-\widehat n(A)
    = \max\left\{\frac{v(A)}{2},v(A)-1\right\} ;
\]
\item[(iii)] optimized fee revenue is
\[
    \widehat Q(A)=\widehat p(A)\widehat n(A)
    =
    \begin{cases}
    \dfrac{v(A)^2}{4}, & v(A)\le2,\\[1em]
    v(A)-1, & v(A)\ge2.
    \end{cases}
\]
If $v$ is differentiable at $A$, then $\widehat Q'(A)=\widehat n(A)v'(A).$
\end{enumerate}

\noindent Suppose $\widehat\Pi(A)=\widehat Q(A)-c_FA$ attains a maximum on $\mathbb R_+$.
The set of platform-optimal activity levels under weak identity is $\mathcal A^* = \arg\max_{A\ge0}\{\widehat Q(A)-c_FA\}.$
No airdrop is optimal, $0\in\mathcal A^*$, if and only if $\widehat Q(A)-\widehat Q(0)\le c_FA$ for every $A>0$.
If $A^*\in\mathcal A^*$ is positive and interior, and $v$ is differentiable at $A^*$, then $\widehat n(A^*)v'(A^*)=c_F.$
\end{lemma}

\begin{proof}
In the appendix.
\end{proof}

We assume that $\widehat\Pi$ attains a maximum on $\mathbb R_+$.
This is a mild existence condition.
It holds, for example, if $v$ is bounded above, or more generally if $\widehat Q(A)-c_FA\to-\infty$ as $A\to\infty$.

\cref{res:optimal_pricing_activity} gives the global profitability test for an airdrop under weak identity.
A positive airdrop is profitable if and only if some activity level raises optimized fee revenue by more than its procurement cost, i.e., there exists an activity level $A>0$ such that $\widehat Q(A)-\widehat Q(0)>c_FA$.
Equivalently, $\sup_{A>0} \frac{\widehat Q(A)-\widehat Q(0)}{A} >c_F.$
The marginal condition $\widehat n(A^*)v'(A^*)=c_F$ says that the platform buys activity until the marginal improvement in organic users' gross value, $v'(A^*)$, multiplied by the monetized mass of organic users, $\widehat n(A^*)$, equals the farmer unit cost.
Below the adoption cap, $\widehat n(A^*)=\frac{v(A^*)}{2}$, so an additional unit of activity both raises willingness to pay and expands adoption.
At the cap, $\widehat n(A^*)=1$, so the extensive margin is exhausted and activity is valuable only because it allows the platform to charge a higher full-adoption fee.

Given any selected $A^*\in\mathcal A^*$, the associated optimal fee and adoption level are $p^*=\widehat p(A^*)$ and $n^*=\widehat n(A^*)$.
If $A^*=0$, the no-airdrop policy is implemented by setting $r=\Omega=0$.
If $A^*>0$, a pure pooled implementation sets $r^*=0$ and $\Omega^*=c_FA^*$.
Generally, any $(r^*,\Omega^*)$ satisfying $r^*<c_F$ and $\Omega^*=(c_F-r^*)A^*$ implements the same activity level and the same total transfer cost $T^*=c_FA^*$.

\subsection{Farmer exclusion and organic-user surplus}
\label{sec:screening_and_surplus}

We now compare weak identity with a counterfactual in which farmers can be perfectly screened out.
The comparison in this subsection uses the source-neutral benchmark: a unit of rewarded activity raises organic users' gross platform value by the same amount whether it is supplied by farmers or by verified organic agents.
\cref{sec:lower_value_farmer_activity} relaxes this benchmark.

Fix an activity level $A$.
Conditional optimal pricing gives $\widehat p(A)=v(A)-\widehat n(A)$.
Therefore, an organic adopter's platform-consumption payoff reduces to $v(A)-\widehat p(A)-\theta=\widehat n(A)-\theta$.
Taking the outside option into account, type $\theta$ receives $u_\theta(A)=\max\{\widehat n(A)-\theta,0\}$.

More source-neutral activity weakly raises this payoff for every organic type, since $\widehat n(A)=\min\{v(A)/2,1\}$ is weakly increasing in $A$.
Thus, organic users benefit even when not rewarded.

To isolate the cost channel, use a common feasible activity set $\mathcal A\subseteq\mathbb R_+$ with $0\in\mathcal A$, and suppose the maximization problems below attain solutions.
Define
\[
    \Pi_F^*
    =
    \max_{A\in\mathcal A}\{\widehat Q(A)-c_FA\},
    \qquad
    A_F^*\in
    \arg\max_{A\in\mathcal A}\{\widehat Q(A)-c_FA\},
\]
for the weak-identity regime with farmers, and
\[
    \Pi_O^*
    =
    \max_{A\in\mathcal A}\{\widehat Q(A)-c_OA\},
    \qquad
    A_O^*\in
    \arg\max_{A\in\mathcal A}\{\widehat Q(A)-c_OA\},
\]
for a verified organic-procurement benchmark in which farmers are screened out and the same source-neutral activity must be procured from verified organic agents at unit cost $c_O$.

\begin{proposition}[Perfect farmer screening can make both sides worse off]
\label{res:farmer_exclusion_pareto}
In the source-neutral benchmark with common feasible set $\mathcal A$,
\[
    \Pi_F^*\ge \Pi_O^*,
    \quad \text{ and } \quad
    A_F^*\ge A_O^*.
\]
Consequently, $u_\theta(A_F^*)\ge u_\theta(A_O^*)$ for every organic type $\theta\in[0,1]$.
Thus, perfect farmer screening weakly reduces both platform profit and every organic user's payoff.
\end{proposition}

\begin{proof}
In the appendix.
\end{proof}

Under the homogeneous-cost baseline, in which every farmer-created claim costs \(c_F\) and every organic claim costs \(c_O>c_F\), \cref{res:farmer_exclusion_pareto} compares two clear-cut cases: farmers are either available and supply the rewarded activity, or they are perfectly screened out and the same activity must be procured from organic users.
Appendix~\ref{sec:primary_organic_advantage} relaxes this assumption by allowing primary organic claims to be cheaper than farmer claims, while additional organic identities are more costly.
A small pool then attracts only primary organic claims, whereas a larger pool recruits farmers to supply activity beyond the inexpensive organic margin.
Pool size therefore selects claimant composition and can screen farmers endogenously.

These results clarify why farmer participation can be valuable even though farmers have no standalone consumption value.
The key is the separation between who is paid by the reward program and who benefits from the activity it procures.
The platform does not care about the identity performing the rewarded action --- it cares whether the action raises the value of the platform for organic users.
Farmers receive the transfer, but competitive entry drives their net payoff to zero and leaves the platform buying the rewarded activity at unit cost $c_F$.
Organic users do not receive the reward directly, but they benefit because the rewarded activity raises platform quality.
The platform then monetizes part of this improvement through its fee, while organic users benefit through weakly higher adoption and higher inframarginal utility.

\cref{res:farmer_exclusion_pareto} formalizes the separation between who is paid by the reward program and who benefits from the activity it procures.
Farmers receive the reward, but in equilibrium they earn no rents; organic users do not receive the reward, but they may receive the quality improvement.
Screening farmers out may look more targeted, but it removes the lowest-cost suppliers of the separable activity.
If that activity must instead be procured from verified organic agents at cost $c_O>c_F$, the platform buys weakly less activity.
In the source-neutral benchmark, that reduces platform profit and weakly lowers organic users' platform-consumption surplus.

The result shows that weak identity can make farmer exclusion inefficient when three conditions hold: the rewarded task is separable from regular use, the task creates platform value for organic users, and farmers are the low-cost suppliers of the rewarded task.
When these conditions fail, screening remains the natural response.

If perfect screening also limits verified organic procurement to at most one unit per organic user, then the screened feasible set is $[0,1]$ while the farmer-enabled feasible set may be larger.
\cref{res:farmer_exclusion_pareto} abstracts from that capacity advantage by imposing a common feasible set.
Allowing farmers to expand the feasible activity set strengthens the case in which screening backfires.\\

\noindent {\bf Example} (Farmer screening failures).
The preceding results are nonparametric.
This example illustrates the procurement-cost comparison using a value function with finite marginal returns: $$v(A)=v_0+\alpha(1-e^{-\psi A}), \qquad v_0>0,  \quad\alpha>0, \quad \psi>0.$$
Here, $v_0$ is standalone platform value, $\alpha$ is the maximum incremental value generated by rewarded activity, and $\psi$ governs the rate at which marginal returns decline.
If $v_0\ge \alpha$, then $\widehat Q'(A)$ is weakly decreasing.\footnote{
  By \cref{res:optimal_pricing_activity}, $\widehat Q'(A)=\frac{v(A)v'(A)}{2}$ when $v(A)<2$, and $\widehat Q'(A)=v'(A)$ when $v(A)>2$.
  Since $v'(A)=\alpha\psi e^{-\psi A}$, the full-adoption expression is decreasing.
  The partial-adoption expression is proportional to $e^{-\psi A}[v_0+\alpha(1-e^{-\psi A})]$, which is decreasing when $v_0\ge \alpha$.}
Hence, a positive activity level is optimal at unit procurement cost $c$ if and only if $c<\widehat Q'(0)$, with any interior optimum characterized by $\widehat Q'(A^*)=c$.

This yields two ways in which farmer screening can backfire.
If $$c_F<\widehat Q'(0)\le c_O,$$ the activity is worth procuring at farmer cost but not at the screened organic cost.
Screening therefore eliminates the rewarded activity channel.
If instead $$c_F<c_O<\widehat Q'(0),$$ both regimes procure positive activity, but screening raises the unit procurement cost and therefore reduces activity and platform quality.

To illustrate, set
\[
    v(A)=1.2+1.2(1-e^{-A}).
\]
Then \(v(0)=1.2\), so initial adoption is unsaturated,
\[
    \widehat Q'(0)=\frac{v(0)v'(0)}{2}=0.72,
    \qquad
    \widehat n(0)=0.60.
\]
Adoption saturates at \(A_{\mathrm{sat}}=\log 3\), where \(\widehat Q'(A_{\mathrm{sat}})=0.40\).
The screened and farmer-enabled optima are easy to compute from the first-order condition \(\widehat Q'(A)=c\).
In the partial-adoption region, \(\widehat Q'(A)=1.44e^{-A}-0.72e^{-2A}\).
In the full-adoption region, \(\widehat Q'(A)=1.2e^{-A}\).

\medskip \noindent\emph{Case A: screening eliminates activity.}
Let \(c_O=0.80\) and \(c_F=0.20\).
Since
\[
    c_F=0.20<\widehat Q'(0)=0.72<c_O=0.80,
\]
the screened benchmark procures no additional activity:
\[
    A_O^*=0,
    \qquad
    \widehat n(A_O^*)=0.60.
\]
With farmers, the lower unit cost makes activity profitable.
Since \(c_F<\widehat Q'(A_{\mathrm{sat}})\), the farmer-enabled optimum lies in the full-adoption region and solves \(1.2e^{-A}=0.20\).
Thus
\[
    A_F^*=\log 6\approx 1.79,
    \qquad
    \widehat n(A_F^*)=1.
\]
Screening farmers out therefore destroys an activity channel that would have raised the platform from partial to full adoption.

\medskip \noindent\emph{Case B: screening reduces activity.}
Let \(c_O=0.64\) and \(c_F=0.16\).
Now
\[
    c_F=0.16<c_O=0.64<\widehat Q'(0)=0.72,
\]
so both regimes procure positive activity.
But the screened benchmark procures less because its unit cost is higher.
The screened optimum lies in the partial-adoption region and solves
\[
    1.44e^{-A}-0.72e^{-2A}=0.64,
\]
which gives
\[
    A_O^*=\log(3/2)\approx0.41,
    \qquad
    \widehat n(A_O^*)=0.80.
\]
With farmers, the lower unit cost again pushes the optimum into the full-adoption region:
\[
    A_F^*=\log(7.5)\approx2.01,
    \qquad
    \widehat n(A_F^*)=1.
\]
Thus screening does not eliminate the activity channel, but it weakens it enough to lower platform quality.

At the platform-optimal fee, organic type $\theta$'s payoff can be written as $u_\theta=\max\{\widehat n-\theta,0\}$.
\Cref{fig:screening_example} plots these payoffs.
In Case A, screening reduces organic payoffs from $\max\{1-\theta,0\}$ to $\max\{0.60-\theta,0\}$.
In Case B, screening reduces them to $\max\{0.80-\theta,0\}$.
In both cases, every organic user $\theta<1$ is strictly better off with farmers, even though the rewards are paid to farmers rather than to organic users.
Farmers are valuable here not because they are consumers, but because weak identity makes them low-cost suppliers of a separable input into platform quality.

\begin{figure}[!ht]
\centering
\begin{tikzpicture}
\begin{groupplot}[
  group style={group size=2 by 1, horizontal sep=1.4cm},
  width=0.43\textwidth,
  height=0.33\textwidth,
  xmin=0, xmax=1,
  ymin=0, ymax=1.05,
  domain=0:1,
  samples=200,
  axis lines=left,
  xlabel={Organic type $\theta$},
  ylabel={Payoff $u_\theta$},
  xtick={0,0.2,0.4,0.6,0.8,1},
  ytick={0,0.2,0.4,0.6,0.8,1},
  ticklabel style={font=\footnotesize},
  label style={font=\footnotesize},
  title style={font=\footnotesize},
  every axis plot/.append style={semithick},
  legend style={draw=none, fill=none, font=\footnotesize}
]

\nextgroupplot[
  title={Case A: screening eliminates activity},
  legend pos=north east
]
\addplot {max(1-x,0)};
\addlegendentry{Weak identity}
\addplot[dashed] {max(0.6-x,0)};
\addlegendentry{Screening}

\nextgroupplot[
  title={Case B: screening reduces activity},
  ylabel={}
]
\addplot {max(1-x,0)};
\addplot[dashed] {max(0.8-x,0)};

\end{groupplot}
\end{tikzpicture}
\caption[Organic user payoffs with and without farmer screening]{Organic user payoffs with and without farmer screening.
The payoff of an organic type $\theta$ is $u_\theta=\max\{\widehat n-\theta,0\}$.
The value function is $v(A)=1.2+1.2(1-e^{-A})$.
The solid line is the weak-identity equilibrium with farmer-supplied activity, while the dashed line is the screened benchmark.
In Case A, $c_O=0.80, c_F=0.20$, so screening eliminates activity and leaves adoption at $\widehat n=0.60$.
In Case B, $c_O=0.64, c_F=0.16$, so screening preserves the activity channel but yields lower adoption, $\widehat n=0.80$.
In both, weak identity with farmers weakly dominates farmer screening.
}
\label{fig:screening_example}
\end{figure}

\section{Model robustness and extensions}
\label{sec:robustness_extensions}

The benchmark deliberately strips down the mechanism.
It assumes that the rewarded task is source-neutral, that the reward is attached to a separable activity rather than to regular access, and that the value of activity is contemporaneous.
This section relaxes those choices.
The extensions preserve the procurement logic but clarify its scope.

\subsection{Why a fee cut is not equivalent to an activity reward}
\label{sec:fee_discount}

The procurement interpretation depends on the reward being attached to a separable platform-building task.
A reward for regular use is different.
If the reward is paid merely for adopting, signing up, logging in, or making regular use of the platform, and that action does not create an independent platform-quality improvement, the reward functions as a price discount: it changes the net access price on the baseline demand curve rather than shifting the demand curve.

To see this, consider a counterfactual regular-use reward.
Let $X>0$ denote the mass of claimants, including organic users and any farmer-created identities that mimic regular use, and let the regular-use reward be
\[
    \sigma(X)=r+\frac{\Omega}{X}.
\]
Because the reward is paid for regular use rather than a separable platform-building task, the activity state remains \(A=0\).
An organic user's adoption payoff is
\[
    v(0)-\theta-p+\sigma(X)
    =v(0)-\theta-\widetilde p,
\]
where $\widetilde p\equiv p-\sigma(X)$ is the ``discounted'' net access price.
Conditional on the realized reward, the regular-use reward is therefore just a price discount.
The inverse demand curve remains $\widetilde p=v(0)-n.$

Farmers may still try to enter a regular-use reward program.
But a farmer who mimics regular use must pay the access fee \(p\).
Abstracting from any additional mimicking cost, farmer entry requires \(\sigma(X)\ge p\).
If farmers are active in a bounded equilibrium, dilution and free entry imply
\[
    \sigma(X)=p.
\]
The platform's profit is then
\[
  pX-\sigma(X)X=0.
\]
Since \(v(0)>0\), the platform can earn positive profit by charging a regular access fee without inducing farmer entry.
Thus, in the regular-use reward regime, it is not optimal to set the reward pool so high that farmers enter.
If farmers are deterred, the policy is equivalent to choosing a lower net access price directly.

By contrast, the baseline activity reward implements $A>0$ by setting $r<c_F$ and $\Omega=(c_F-r)A.$
Then, farmer free entry gives $\tau(A)=c_F$, farmers supply the rewarded activity, and organic users do not claim.
Organic users face the shifted inverse demand curve $p_A(n)=v(A)-n$ instead of the baseline-quality curve $p_0(n)=v(0)-n.$
For any fixed adoption level $n$, activity raises willingness to pay by $p_A(n)-p_0(n)=v(A)-v(0).$
The activity reward therefore changes the demand curve by procuring an input that raises platform value.

This distinction clarifies the scope of the model.
Rewards for regular platform use are not covered by the main result unless the rewarded use creates an independent platform-quality benefit.
If trading-volume rewards merely rebate fees or induce wash trading, they are closer to a price discount.
If the rewarded action creates genuine liquidity, price discovery, listings, content, data, or other separable value for organic users, that incremental value is represented by $A$ and $v(A)$.

\subsection{When farmer-supplied activity is lower value}
\label{sec:lower_value_farmer_activity}

The source-neutral benchmark isolates the procurement-cost channel.
It suits tasks with outputs that are valuable irrespective of their sources.
For example, all else being equal, liquidity from organic users and farmers is identical.
In other cases, the source may matter: consider how reviews or governance participation may be more valuable when generated by organic users than by farmers.

Let $v_F(A)$ denote organic users' gross platform value when activity $A$ is supplied by farmers, and let $v_O(A)$ denote gross platform value when the same amount of activity is supplied by verified organic agents.
Assume $v_F$ and $v_O$ are increasing and concave, with $v_F(0)=v_O(0)>0,$ and $v_O(A)>v_F(A)$ for every $A>0$.
The benchmark model is the special case $v_F=v_O=v$.

For $s\in\{F,O\}$, define source-specific conditional adoption, price, and fee revenue by
\[
    \widehat n_s(A)=\min\left\{\frac{v_s(A)}{2},1\right\},
    \qquad
    \widehat p_s(A)=v_s(A)-\widehat n_s(A),
\]
and
\[
    \widehat Q_s(A)
    =
    \begin{cases}
    \dfrac{v_s(A)^2}{4}, & v_s(A)\le2,\\[1em]
    v_s(A)-1, & v_s(A)\ge2.
    \end{cases}
\]
Because $v_F(0)=v_O(0)$, no-airdrop revenue is the same in the two regimes: $\widehat Q_F(0)=\widehat Q_O(0).$

The farmer-supplied airdrop is profitable relative to no airdrop if and only if there exists $A>0$ such that
\[
    \widehat Q_F(A)-\widehat Q_F(0)>c_FA,
\]
or equivalently,
\[
    \sup_{A>0}
    \frac{\widehat Q_F(A)-\widehat Q_F(0)}{A}>c_F.
\]
Thus farmer-supplied activity need not be as valuable as organic-supplied activity.
It only needs to create enough monetizable platform value to cover the farmer procurement cost.

The comparison with perfect screening is a cost-quality tradeoff.
Let
\[
    \Pi_F^*=
    \max_{A\in\mathcal A}\{\widehat Q_F(A)-c_FA\}
\]
be the weak-identity payoff with farmer-supplied activity, and let
\[
    \Pi_O^*=
    \max_{A\in\mathcal A}\{\widehat Q_O(A)-c_OA\}
\]
be the verified organic-procurement payoff on a common feasible set $\mathcal A$.
If $A_O^*$ is an optimizer of the verified benchmark, then perfect farmer screening backfires for the platform whenever
\[
    (c_O-c_F)A_O^*
    \ge
    \widehat Q_O(A_O^*)-\widehat Q_F(A_O^*).
\]
The left-hand side is the procurement-cost saving from using farmers at the screened benchmark's activity level.
The right-hand side is the optimized fee-revenue loss from the lower value of farmer-supplied activity.
Screening backfires for the platform when the cost saving exceeds this quality loss.

We now turn to the user-surplus comparison.
Let $A_F^*$ and $A_O^*$ be selected optimizers in the two regimes.
Farmer availability weakly Pareto-improves organic users relative to perfect screening if and only if $\widehat n_F(A_F^*)\ge \widehat n_O(A_O^*).$
When adoption is unsaturated in both regimes, this condition is equivalent to $v_F(A_F^*)\ge v_O(A_O^*).$
With source-neutral activity, lower farmer cost implies weakly higher procured activity and weakly higher organic-user surplus.
With lower-value farmer activity, that monotone comparison need not hold.
The platform may prefer farmers because they are cheaper, while organic users may prefer screening because verified organic activity creates higher total platform value.
The user comparison depends on induced platform quality, summarized by $\widehat n_s(A_s^*)$, not on procurement cost by itself.

This qualification does not overturn the procurement mechanism.
Weak identity still makes the platform buy the rewarded input from the low-cost claimants.
The additional point is that low-cost procurement is valuable only relative to the quality of the input procured.
Farmer-supplied activity can be profitable, and perfect screening can still backfire, but the comparison becomes a cost-quality tradeoff rather than a pure cost comparison.

\paragraphN{Selective screening}
The benchmark comparison is between perfect source-based exclusion and access to farmer supply; it does not imply that platforms should ignore claim-level signals.
For example,  highly correlated, duplicative, mechanically generated, congestive, or governance-distorting claims may create less effective activity or impose direct costs.
Screening such claims can raise input quality, while excluding every farmer can also remove useful low-cost supply.
The relevant policy can therefore be partial: exclude claims whose expected monetizable contribution is below their procurement and associated harm costs, while retaining farmer-supplied claims that create sufficient net value.
\Cref{res:pool_composition} adds an incentive-screening margin: with cheaper primary organic claims, a small pool excludes farmers altogether, whereas a larger pool recruits them only when the platform values expansion beyond primary organic capacity.

\paragraphN{Token rewards and post-airdrop sales}
The baseline measures rewards by the value claimants expect to receive.
It therefore takes as given the conversion from a token package into claimant value and studies procurement conditional on that value.
For a pool specified as a number of tokens, however, claimant value may depend on the number of tokens distributed and on anticipated post-airdrop sales.
The identity of recipients may also matter after the rewarded task has been completed if token ownership affects, e.g., governance, alignment, coordination, or organic users' willingness to use the platform.
Existing evidence documents rapid post-airdrop selling and exit, but does not by itself identify the magnitude or even the sign of these downstream effects \parencite{jhackworth2022uniswap,fan2023altruistic,messias2023airdrops}.

To the extent that farmer receipt or sale of tokens changes organic users' platform value, the source-specific comparison in this subsection can incorporate it by interpreting \(v_F(A)\) as the net value of farmer-supplied activity, including these downstream effects.
If the effect operates only through gains or losses on existing token holdings, it affects broader token-holder welfare rather than the platform-consumption payoff studied here.
Farmer participation in a token airdrop is therefore desirable only when its procurement-cost advantage is large enough to offset both any lower value of farmer-supplied activity and any disadvantage created by the resulting token allocation.
The pool-dilution logic continues to apply conditional on claimant value.
Determining how a fixed number of tokens maps into that value would require an additional model of token-price formation.

\subsection{Multiple periods and durable platform value}
\label{sec:dynamic}

The baseline model treats rewarded activity as contemporaneously valuable through $v(A)$.
The same procurement logic applies when rewarded activity also creates a durable platform stock.
This extension addresses settings where the rewarded task builds liquidity, content, listings, reputation, data, documentation, or integrations that remain useful after the airdrop has been paid.

Consider two periods.
The platform chooses a single airdrop in period 1.
As in the baseline, any positive activity level $A$ implemented under weak identity is supplied by farmers and costs $c_FA$.
The activity generates period-1 value $v(A)$ and leaves a residual stock $\rho A$ in period 2, where $\rho\in[0,1]$ measures persistence.
Period-2 platform value is therefore $v(\rho A)$.
The platform discounts period-2 profit by $\beta\in(0,1)$ and chooses its access fee optimally in each period.

Let
$$\widehat Q_D(A)
  \equiv
  \widehat Q(A)+\beta\widehat Q(\rho A)$$
denote total optimized fee revenue across the two periods.
The platform's two-period optimized payoff from period-1 activity $A$ is $$\widehat\Pi_D(A)=\widehat Q_D(A)-c_FA.$$
The durable activity choice is therefore $\mathcal A_D^* = \arg\max_{A\ge0}\{\widehat Q_D(A)-c_FA\},$ assuming the maximum exists.
No airdrop is globally optimal if and only if $\widehat Q_D(A)-\widehat Q_D(0)\le c_FA$ for every $A>0$.
For any selected $A_D^*\in\mathcal A_D^*$, the period-specific optimized fees are
\[
    p_1^*=\widehat p(A_D^*),
    \qquad
    p_2^*=\widehat p(\rho A_D^*).
\]
If $A_D^*>0$ is an interior optimizer and $v$ is differentiable at $A_D^*$ and $\rho A_D^*$, then the first-order condition is
\[
    \widehat Q'(A_D^*)+
    \beta\rho\widehat Q'(\rho A_D^*)
    =
    c_F.
\]
The total airdrop transfer is $c_FA_D^*$.

The multi-period environment changes the marginal return, not the implementation logic.
The platform still buys activity from farmers at unit cost $c_F$.
The difference is that the same unit of activity can now be monetized in two periods.
The term $\widehat Q'(A_D^*)$ is the period-1 marginal revenue effect.
The term $\beta\rho\widehat Q'(\rho A_D^*)$ is the discounted marginal value of the residual stock in period 2.
Thus persistence expands the set of environments in which airdrops are profitable.
Any airdrop that is profitable in the one-period model remains profitable in the durable model, and persistence can make an airdrop profitable even when the one-period condition fails.

When $\rho=0$, the durable extension collapses to the static problem plus a constant continuation payoff.
In a repeated version of the model with a fresh airdrop choice each period and no persistence, the platform's problem is separable across periods.
Hence, if a positive airdrop is optimal in the one-period model, it is optimal to issue a positive airdrop in every period.

\section{Conclusion}
\label{sec:conclusion}

Platforms increasingly use rewards to build markets, but weak identity makes targeting imprecise.
A platform may intend to reward distinct users (e.g., sellers, reviewers, or liquidity providers), but it may observe only claim records (e.g., reviews or referrals) rather than identities.
The standard design response is to treat every reward captured by duplicated accounts, bonus hunters, or airdrop farmers as leakage.
Although such a response is appropriate in some settings, our analysis shows that applying it indiscriminately can be counterproductive.
The economic effect of weak identity depends on what eligibility rules reward.

The paper identifies a procurement channel.
When a reward is tied to a separable platform-building task, weak identity can convert a user subsidy into identity-blind procurement.
A shared reward pool is central to this mechanism because it makes claims rivalrous: additional claims dilute the payment received by all claimants, and competitive entry drives the realized per-claim reward down to marginal claim cost.
If reward-seeking claimants are the low-cost suppliers of eligible claims, the platform procures the task from them rather than from the organic users it ultimately wants to serve.
The transfer goes to farmers, but the resulting platform-quality improvement accrues to organic users and can be monetized by the platform through access fees.
In the source-neutral benchmark, this force is strong enough that perfect farmer screening can reduce platform profit and weakly lower every organic user's platform-consumption payoff.
The farmer advantage need not extend to primary organic claims: when additional organic identities are more costly, farmers can enter only after a larger pool has exhausted the cheaper primary-organic margin.

The result is not a general defense of fake activity or reward farming.
Its value rests on three conditions.
First, the rewarded action must create platform value that is separable from the claimant's continued use; otherwise the reward is essentially a price discount or a transfer to claim creation.
Second, the reward architecture must control endogenous claim creation, for example through a shared pool, rather than paying a non-rival per-claim reward that invites unbounded entry.
Third, farmer-supplied activity must create enough monetizable value relative to its procurement cost.
When farmer output is of lower value than organic-user output, the comparison becomes a cost-quality tradeoff.
When the activity is temporary, the return depends on how long the platform can monetize the procured input.
Under competition, low-cost procurement can benefit users while dissipating platform profits through a costly quality contest.

These qualifications lead to a practical design principle.
Platforms should not ask only whether rewards are captured by the intended identities.
They should ask what the eligibility rule actually buys.
Liquidity, capital, inventory, content, reviews, data, documentation, and integrations may be valuable even when supplied by claimants who do not remain long-run users.
Signups, logins, mechanical transactions, fee cycling, and wash trades often are not.
The appropriate response to weak identity is therefore not always stricter screening and not always greater tolerance.
It is task-specific screening: exclude claimants and behaviors that create no independent value or impose costs on other users, but recognize that low-cost claimants can be useful suppliers when the rewarded task itself builds the platform and the reward schedule disciplines entry.

The baseline model gives the core operational guidance.
A shared reward pool controls claim creation by making claims rivalrous, and, when farmer cost is known, its size determines how much platform-building activity the platform procures.
Screening should then depend on the net value of that activity rather than simply on whether the claimant is a farmer: perfect exclusion can be counterproductive when farmers are low-cost suppliers of a valuable task.
When farmer cost is uncertain, the platform cannot guarantee a particular activity level ex ante, but can target expected activity through its choice of \(\Omega\).

The additional results refine this guidance.
Under the cost ordering in \cref{res:pool_composition}, pool size is also a sourcing instrument.
In particular, it should be kept below the farmer-entry threshold when inexpensive primary organic supply is sufficient, and above that threshold only when the incremental value of expansion covers the activation cost.
For token airdrops, managers should distinguish between procuring the rewarded task and allocating the platform's tokens.
The lowest-cost supplier of the task need not be the most desirable token recipient.
Farmer participation should therefore be evaluated by weighing its procurement-cost advantage against any lower value of farmer-supplied activity and any adverse effects of the resulting token allocation or anticipated post-airdrop sales.
When those effects are important, selective screening or compensating the task with an instrument not tied to long-term token ownership becomes a relevant design option.
The model does not rank these alternatives; which one is preferable depends on empirical magnitudes that the analysis identifies but does not estimate.

The analysis also clarifies what must be measured empirically.
Post-airdrop retention among recipients and the share of rewards paid to suspected farmers are important, but they measure who receives the transfers and what recipients do shortly after claiming.
They do not, by themselves, determine whether the reward program created value for other users.
For platform managers and empirical researchers, the central quantities are the incremental value of the rewarded activity, its persistence, the quality difference between farmer- and organic-supplied output, the marginal cost of procuring that activity from different claimant pools, and, for token airdrops, the effects of the resulting token allocation and anticipated sales.

The broader message extends beyond airdrops.
Referral bonuses, seller subsidies, review rewards, content bounties, data-labeling incentives, and contributor programs all face versions of the same implementation problem: platforms often observe claim records more easily than they observe distinct valuable individuals.
When rewards are attached to claims under weak identity, incentive design determines whether the program buys growth, buys a useful platform input, or simply buys manipulation.
The contribution of this paper is to show that these cases should not be treated alike.
Weak identity can make reward programs fail through leakage, but it can also make them work through low-cost procurement.
Distinguishing these cases is essential for designing platform subsidies in markets where identities are easy to create but valuable activity is costly to produce.

\appendix
\section*{Appendix}

\section{Heterogeneous claim costs and multi-accounting}
\label[appendix]{sec:heterogeneous}

The baseline model assumes that all organic users face the same marginal cost $c_O$ of performing the rewarded action, while all farmer-created identities face the same marginal cost $c_F<c_O$.
We now relax that homogeneous-cost and one-account structure.
Marginal claim costs may vary across persons and across successive accounts, and regular users may also create multiple eligible identities when doing so is privately profitable.
We first use this structure to study the alternative ordering in which primary organic claims are cheaper than farmer claims, while additional organic identities are more costly.

The extension preserves the main economic distinction between organic users and farmers.
Organic users are demand-side agents: they value platform access and may also claim rewards.
Farmers are not demand-side agents: they obtain no standalone consumption value from the platform and participate only to earn the reward net of claim costs.
Thus an organic user who creates multiple eligible identities remains an organic user, while a farmer remains a pure reward claimant.
The rewarded action remains a separable platform-building input: total rewarded activity $A$ raises organic users' gross platform value from $v(0)$ to $v(A)$.

\subsection{Claim technologies and mixed claiming}

Organic users are indexed by $\theta\in[0,1]$, as in the main model.
The index $\theta$ is the user's demand-side match type.
For each organic user $\theta$, let $m_\theta^O(j)$ denote the incremental net cost of creating the $j$-th eligible identity and performing the rewarded action through that identity.
This cost may include the operational cost of satisfying eligibility requirements, the cost of creating and managing an additional account, detection risk, moral costs, and any private taste for receiving rewards.
We assume only that successive identities are weakly more costly for the same person: $m_\theta^O(j+1)\ge m_\theta^O(j)$ for every $\theta\in[0,1]$ and $j\ge1$.
No monotonicity is imposed across organic users.

For a per-identity reward $\tau$, organic user $\theta$'s eligible identity count is
$$
k_\theta^O(\tau)
=
\sum_{j\ge1}\mathbf 1\{m_\theta^O(j)\le\tau\},
$$
and the claim surplus, conditional on adopting the platform, is
\[
    S_\theta^O(\tau)
    =
    \sum_{j\ge1}[\tau-m_\theta^O(j)]_+.
\]
We assume that $k_\theta^O(\tau)$ and $S_\theta^O(\tau)$ are finite and integrable for every finite $\tau$.
If organic user $\theta$ adopts, the payoff is $v(A)-\theta-p+S_\theta^O(\tau).$
While the user's additional eligible identities contribute to rewarded activity $A$, they do not increase the mass $n$ of organic adopters.

Farmers are represented by an unbounded supply of potential marginal farmed identities.
Let $(\mathcal X_F,\nu_F)$ be the space of potential farmer-created marginal identities, where $\nu_F$ is sigma-finite and $\nu_F(\mathcal X_F)=\infty$.
A marginal farmed identity $x\in\mathcal X_F$ has claim cost $m_F(x)$.
The aggregate farmer identity supply function is
$$F_F(\tau)
=
\nu_F\{x\in\mathcal X_F:m_F(x)\le\tau\}.$$
We assume local finiteness, $F_F(\tau)<\infty$, and continuity of $F_F$ at every finite $\tau$, as well as unbounded potential supply, $\lim_{\tau\to\infty}F_F(\tau)=\infty.$
Thus the extension does not bound the number of potential farmer-created identities.
A finite reward induces a finite mass of farmed identities because only finitely many marginal identities have cost below that reward.

For later reference, define the maximum potential organic marginal-identity supply
$$\bar F_O(\tau)
=
\int_0^1 k_\theta^O(\tau)\,d\theta.$$
This is the mass of organic marginal identities that would claim at reward $\tau$ if all organic users adopted.
In equilibrium, the actual mass of organic rewarded identities may be smaller because only adopters can claim.

Given an airdrop policy $(r,\Omega)$, the per-identity reward is $\tau(A)=r+\frac{\Omega}{A}$ for $A>0.$
A bounded mixed-claim equilibrium with heterogeneous marginal claim costs is a tuple $(n,a_O,a_F,A,\tau)$ such that:
\begin{align*}
\tau
&
=
r+\frac{\Omega}{A}
\\
n
&
=
\int_0^1
\mathbf 1\{v(A)-\theta-p+S_\theta^O(\tau)\ge0\}
\,d\theta
\\
a_O
&
=
\int_0^1
\mathbf 1\{v(A)-\theta-p+S_\theta^O(\tau)\ge0\}
k_\theta^O(\tau)\,d\theta
\\
a_F
&
=
F_F(\tau)
\\
A
&
=
a_F+a_O
\end{align*}
This system allows users to farm at the margin.
Some organic users may create multiple eligible identities, and those identities receive the same reward as farmer-created identities.
At the same time, farmers remain pure reward claimants and do not enter the organic demand term $n$.

\subsection{Primary-account advantage and farmer expansion}
\label{sec:primary_organic_advantage}

The general claim technology allows organic users to have an advantage on their primary accounts even when farmers have an advantage in producing additional eligible identities.
This subsection that under such ordering the size of a shared pool can screen farmers endogenously.

Fix an access fee and a reward range over which every organic user adopts.
Let
\[
    G_1(\tau)
    =
    \int_0^1\mathbf 1\{m_\theta^O(1)\le\tau\}\,d\theta
\]
denote the supply of primary organic claims at reward $\tau$.
Suppose primary claim costs have a continuous distribution with upper support $\bar c_O^1$, so $G_1(\bar c_O^1)=1$.
Let
\[
    \underline c_F
    =
    \inf\{\tau\ge0:F_F(\tau)>0\}
\]
be the lower support of farmer claim costs, and let
\[
    \underline c_O^+
    =
    \inf_{\theta\in[0,1],\,j\ge2}m_\theta^O(j)
\]
be the lowest cost of an additional organic identity.
Assume
\[
    \bar c_O^1<\underline c_F<\underline c_O^+.
\]
Thus, a primary organic claim is cheaper than the cheapest farmer claim, but a farmer claim is cheaper than every additional organic identity.
For some $\bar\tau\in(\underline c_F,\underline c_O^+)$, suppose $F_F(\underline c_F)=0$, $F_F(\tau)>0$ for $\tau>\underline c_F$, and $F_F$ is strictly increasing on $[\underline c_F,\bar\tau]$.
Finally, suppose
\[
    H(\tau)
    =
    \tau\bigl[G_1(\tau)+F_F(\tau)\bigr]
\]
is strictly increasing on $[0,\bar\tau]$.

\begin{proposition}[Pool size selects claimant composition]
\label{res:pool_composition}
Consider a pure shared pool, $r=0$, with
\[
    0<\Omega\le \bar\Omega
    \equiv
    \bar\tau\bigl[1+F_F(\bar\tau)\bigr].
\]
Within the stipulated full-adoption reward range, there is a unique bounded mixed-claim equilibrium with realized reward $\tau\in(0,\bar\tau]$.
Its composition satisfies:
\begin{enumerate}[label=(\roman*), leftmargin=2em]
\item If $0<\Omega<\bar c_O^1$, then $\tau<\bar c_O^1$, no farmers enter, and rewarded activity consists only of the primary claims whose costs do not exceed $\tau$:
\[
    a_F=0,
    \qquad
    a_O=A=G_1(\tau)<1.
\]
\item If $\bar c_O^1\le\Omega\le\underline c_F$, then all primary organic users claim, no farmers enter, and
\[
    \tau=\Omega,
    \qquad
    a_O=A=1,
    \qquad
    a_F=0.
\]
Thus increasing the pool within this interval raises transfers but not activity.
\item If $\underline c_F<\Omega\le\bar\Omega$, then all primary organic users claim and farmers supply the expansion beyond the primary organic capacity:
\[
    a_O=1,
    \qquad
    a_F=F_F(\tau)>0,
    \qquad
    A=1+F_F(\tau)>1.
\]
No organic user creates an additional identity because $\tau<\underline c_O^+$.
\end{enumerate}

\noindent Equivalently, the minimum uniform-reward transfer needed to implement a target $A$ in the relevant range is
\[
    C_{1F}(A)
    =
    \begin{cases}
    A G_1^{-1}(A), & 0<A\le1,\\[0.7em]
    A F_F^{-1}(A-1), & 1<A\le1+F_F(\bar\tau).
    \end{cases}
\]
In particular,
\[
    C_{1F}(1)=\bar c_O^1,
    \qquad
    \lim_{A\downarrow1,\,A>1}C_{1F}(A)=\underline c_F>\bar c_O^1.
\]
Farmer expansion therefore requires an activation increase in total transfers because the higher uniform reward must also be paid to all inframarginal primary organic claims.
\end{proposition}

\begin{proof}
For every $\tau\le\bar\tau<\underline c_O^+$, no organic user creates an additional identity.
Because every organic user adopts over the stipulated range, organic claim supply is therefore $a_O=G_1(\tau)$, while farmer claim supply is $a_F=F_F(\tau)$.
Hence
\[
    A=G_1(\tau)+F_F(\tau).
\]
Under a pure pool, $\tau=\Omega/A$, so an equilibrium reward solves
\[
    \Omega
    =
    \tau\bigl[G_1(\tau)+F_F(\tau)\bigr]
    =H(\tau).
\]
The assumed strict monotonicity of $H$ gives a unique solution in $(0,\bar\tau]$ for every $\Omega\in(0,\bar\Omega]$.

If $\tau<\bar c_O^1$, then $F_F(\tau)=0$ because $\bar c_O^1<\underline c_F$, while $G_1(\tau)<1$; this gives part (i).
For $\tau\in[\bar c_O^1,\underline c_F]$, we have $G_1(\tau)=1$ and $F_F(\tau)=0$.
Thus $H(\tau)=\tau$, so $\bar c_O^1\le\Omega\le\underline c_F$ is equivalent to $\bar c_O^1\le\tau\le\underline c_F$, proving part (ii).
If $\Omega>\underline c_F=H(\underline c_F)$, strict monotonicity gives $\tau>\underline c_F$.
Then $G_1(\tau)=1$ and $F_F(\tau)>0$, which proves part (iii).

For $0<A\le1$, the least reward that supplies $A$ primary organic claims is $G_1^{-1}(A)$, and all $A$ claims receive that reward, yielding $A G_1^{-1}(A)$.
For $A>1$, all primary organic claims participate and farmer supply must equal $A-1$.
The least uniform reward that elicits this farmer supply is $F_F^{-1}(A-1)$, yielding total transfers $A F_F^{-1}(A-1)$.
Continuity at the lower farmer support gives $F_F^{-1}(A-1)\to\underline c_F$ as $A\downarrow1$ from above, while $G_1^{-1}(1)=\bar c_O^1$.
The transfer discontinuity follows.
\end{proof}

\Cref{res:pool_composition} separates the cost of performing a primary organic task from the cost of expanding the number of eligible identities.
Farmers may be strictly more costly than organic users on every primary account and still become the marginal suppliers when the platform seeks more than one claim per organic adopter.
The result also gives weak identity an endogenous screening role.
A pool at or below $\underline c_F$ excludes farmers without excluding inexpensive primary organic claims; crossing that threshold recruits farmers to expand activity beyond the primary organic capacity.
Because the larger pool also raises the common reward received by inframarginal organic claimants, it changes both platform transfers and organic users' claim surplus, and hence potentially their willingness to pay for access.
Whether farmer expansion is profitable therefore depends on whether the additional platform value created by the expansion, together with any resulting change in access revenue, outweighs the additional transfers required to induce it.
Thus the platform may optimally screen farmers with a small pool or recruit them with a larger pool, depending on whether the value of expansion covers both farmer supply and the activation increase in inframarginal transfers.

\subsection{Lower-tail dominance}

\Cref{res:pool_composition} studies a cost distribution whose lowest-cost claims are primary organic claims.
The condition below covers a complementary case: over the relevant reward range, farmer-created marginal identities are more abundant in the lower tail.

The baseline model assumes $c_F<c_O$.
With heterogeneous marginal costs, the analog is that farmer-created marginal identities are more abundant in the lower tail of the claim-cost distribution.
This does not mean that organic users never claim.
It means that, at the rewards relevant for the airdrop, most low-cost marginal identities are farmer-created.

\begin{assumption}[Farmer lower-tail dominance]
\label{ass:farmer_lower_tail_dominance}
Fix a reward range $\mathcal T$ and a number $q\in(\frac{1}{2},1)$.
Farmers $q$-dominate the lower tail on $\mathcal T$ if $F_F(\tau)\ge \frac{q}{1-q}\bar F_O(\tau)$ for every $\tau\in\mathcal T$.
Define the mixed potential marginal-identity supply curve $F_M(\tau)=F_F(\tau)+\bar F_O(\tau).$
When the relevant inverse functions exist, define $w_M(A)=F_M^{-1}(A)$ and $ w_O(A)=\bar F_O^{-1}(A),$ and the associated uniform-reward procurement costs $C_M(A)=Aw_M(A)$ and $ C_O(A)=Aw_O(A).$
For activity comparisons, the strong form of lower-tail dominance additionally requires that $C_M$ and $C_O$ are differentiable on the relevant activity range and satisfy $C_M'(A)\le C_O'(A)$ on that range.
\end{assumption}

The first part of \cref{ass:farmer_lower_tail_dominance} is the heterogeneous analog of $c_F<c_O$: farmer-created marginal identities dominate the low-cost tail.
It implies that farmers receive most of the airdrop at rewards in $\mathcal T$.
It also implies total transfer-cost dominance in the procurement benchmark.
Indeed, if $\tau=w_O(A)$, then $\bar F_O(\tau)=A$, while $F_M(\tau) = F_F(\tau)+\bar F_O(\tau) \ge A.$
Therefore, the mixed market can supply $A$ rewarded identities at a reward no larger than $w_O(A)$: $w_M(A)\le w_O(A),$ and hence $C_M(A)\le C_O(A).$
The marginal condition $C_M'(A)\le C_O'(A)$ is stronger.
It is needed only for the conclusion that farmer availability leads the platform to procure weakly more activity.
In the homogeneous benchmark, both total and marginal dominance reduce to $c_F<c_O$.

\begin{lemma}[Mixed-claim implementation and transfer shares]
\label{res:heterogeneous_mixed_incidence}
In any bounded mixed-claim equilibrium with $A>0$, total rewarded activity satisfies
$$A
  =
  F_F(\tau)
  +
  \int_0^1
  \mathbf 1\{v(A)-\theta-p+S_\theta^O(\tau)\ge0\}
  k_\theta^O(\tau)\,d\theta.$$
Total transfers are $T(A,\tau)=A\tau.$
Conversely, suppose $A>0$, $\tau>0$, $p$, and $n$ satisfy the above equilibrium system.
Then every pair $(r,\Omega)$ satisfying $0\le r<\tau$ and $\Omega=(\tau-r)A$ implements the same activity level, claimant composition, and total transfer cost.

If \cref{ass:farmer_lower_tail_dominance} holds at the realized reward $\tau$, then farmers receive at least a fraction $q$ of rewarded identities and at least a fraction $q$ of total airdrop transfers: $$\frac{a_F}{A}\ge q \quad \text{ and } \quad \frac{a_F\tau}{A\tau}\ge q.$$
The organic claim mass $a_O$ may nevertheless be strictly positive.
\end{lemma}

\begin{proof}
Given reward $\tau$, farmers supply all marginal identities with $m_F(x)\le\tau$, so $a_F=F_F(\tau)$.
An organic user $\theta$ adopts if and only if
\[
    v(A)-\theta-p+S_\theta^O(\tau)\ge0.
\]
Conditional on adopting, she creates $k_\theta^O(\tau)$ eligible identities.
Aggregating across organic users gives the expression for $a_O$.
Since total rewarded activity is $A=a_F+a_O$, the stated activity equation follows.

Each of the $A$ rewarded identities receives $\tau$, so total transfers are $A\tau$.
If $\Omega=(\tau-r)A$, then
\[
    r+\frac{\Omega}{A}=r+\frac{(\tau-r)A}{A}=\tau,
\]
so any such $(r,\Omega)$ produces the same per-identity reward and therefore the same equilibrium claimant composition and transfer cost.

For the transfer-share claim, the actual organic claim mass is bounded by potential organic supply:
\[
    a_O\le \bar F_O(\tau).
\]
\cref{ass:farmer_lower_tail_dominance} gives
\[
    F_F(\tau)
    \ge
    \frac{q}{1-q}\bar F_O(\tau)
    \ge
    \frac{q}{1-q}a_O.
\]
Since $a_F=F_F(\tau)$ and $A=a_F+a_O$, this is equivalent to $a_F/A\ge q$.
Because every rewarded identity receives the same reward $\tau$, the transfer share equals the identity share.
\end{proof}

\subsection{Platform and user comparisons}

The central economic message survives heterogeneous claim costs.
Even when farmers receive most airdrop transfers, the airdrop can benefit both the platform and organic users because the transfers procure platform-building activity.
The caveat is that organic multi-accounting creates an additional rent-extraction margin: if some organic users also receive heterogeneous claim surplus, an unrestricted profit-maximizing platform may raise the access fee to extract those rents.
The following proposition separates the procurement-cost comparison from this rent-extraction issue.

\begin{proposition}[Heterogeneous claim costs and farmer screening]
\label{res:heterogeneous_screening}
Suppose \cref{ass:farmer_lower_tail_dominance} holds on the relevant reward range.
\begin{enumerate}[label=(\roman*), leftmargin=2em]
\item Consider a mixed-claim equilibrium with activity $A>0$, reward $\tau>0$, and total transfers $A\tau$.
If the platform charges $p=\widehat p(A)$, where $\widehat p(A)$ is computed from the platform-quality effect $v(A)$, then organic adoption is at least $\widehat n(A)$, every organic type is weakly better off than under no airdrop with fee $\widehat p(0)$, and platform profit is at least $\widehat Q(A)-A\tau.$
Hence the platform weakly benefits relative to no airdrop whenever $\widehat Q(A)-\widehat Q(0)\ge A\tau.$
If $\tau\in\mathcal T$, farmers receive at least a fraction $q$ of transfers.

\item Let $\Pi_M^* = \max_A\{\widehat Q(A)-C_M(A)\}$ be the platform's optimized payoff in the mixed-claim procurement benchmark, and let $\Pi_O^* = \max_A\{\widehat Q(A)-C_O(A)\}$ be the optimized payoff in the verified organic-procurement benchmark.
If the maximization problems attain solutions, then $\Pi_M^*\ge \Pi_O^*.$
The inequality is strict whenever the screened benchmark procures positive activity $A_O^*>0$ and $C_M(A_O^*)<C_O(A_O^*)$.

\item If, in addition, the strong form of \cref{ass:farmer_lower_tail_dominance} holds on the relevant activity range, $\widehat Q$ is concave, $C_M$ and $C_O$ are convex, $\widehat Q$, $C_M$, and $C_O$ are differentiable at the screened optimizer $A_O^*$, and the optimizers are unique and interior, then $A_M^*\ge A_O^*.$
Consequently, under the access fee $\widehat p(A)$, organic users' platform-consumption payoffs are weakly higher when farmers are available: $U_\theta^{\mathrm{plat}}(A_M^*)\ge U_\theta^{\mathrm{plat}}(A_O^*)$ for every $\theta\in[0,1]$, where $U_\theta^{\mathrm{plat}}(A)=\max\{\widehat n(A)-\theta,0\}.$
\end{enumerate}
\end{proposition}

\begin{proof}
For (i), under fee $p=\widehat p(A)$, organic user $\theta$'s adoption payoff is $$v(A)-\theta-\widehat p(A)+S_\theta^O(\tau).$$
Using $\widehat p(A)=v(A)-\widehat n(A)$, this becomes $\widehat n(A)-\theta+S_\theta^O(\tau).$
Since $S_\theta^O(\tau)\ge0$, every type $\theta<\widehat n(A)$ has positive adoption payoff.
Hence adoption is at least $\widehat n(A)$.
Type $\theta$'s airdrop payoff is therefore at least $ \max\{\widehat n(A)-\theta,0\}.$
Because $v$ is increasing, $\widehat n(A)$ is weakly increasing in $A$, so this is at least the no-airdrop payoff $ \max\{\widehat n(0)-\theta,0\}.$
Fee revenue is $\widehat p(A)n\ge \widehat p(A)\widehat n(A)=\widehat Q(A)$.
Subtracting total transfers $A\tau$, platform profit is at least $\widehat Q(A)-A\tau$.
The transfer-share claim follows from \cref{res:heterogeneous_mixed_incidence}.

For (ii), \cref{ass:farmer_lower_tail_dominance} implies $C_M(A)\le C_O(A)$ on the relevant activity range.
Let $A_O^*$ solve the verified organic-procurement problem.
Evaluating the mixed-claim problem at $A_O^*$ gives
$$\Pi_M^*
  \ge
  \widehat Q(A_O^*)-C_M(A_O^*)
  \ge
  \widehat Q(A_O^*)-C_O(A_O^*)
  =
  \Pi_O^*.$$
The strict part follows immediately if $A_O^*>0$ and $C_M(A_O^*)<C_O(A_O^*)$.

For (iii), define $G_g(A)=\widehat Q(A)-C_g(A)$ for $g\in\{M,O\}.$
Concavity of $\widehat Q$ and convexity of $C_g$ imply that each $G_g$ is concave.
At an interior screened optimum, $G_O'(A_O^*)=0,$ so $\widehat Q'(A_O^*)=C_O'(A_O^*).$
The strong lower-tail dominance condition gives $C_M'(A_O^*)\le C_O'(A_O^*)$.
Therefore:
\begin{align*}
G_M'(A_O^*)
&
=
\widehat Q'(A_O^*)-C_M'(A_O^*)
\\
&
\ge
\widehat Q'(A_O^*)-C_O'(A_O^*)
\\
&
=
0
.
\end{align*}
$G_M$ is concave and has a unique interior maximizer, so $A_M^*\ge A_O^*$.
As $\widehat n(A)$ weakly increases in $A$, so does the platform-consumption payoff $U_\theta^{\mathrm{plat}}(A)=\max\{\widehat n(A)-\theta,0\}$.
The user comparison follows.
\end{proof}

The farmer-only case from the main model is a useful benchmark, but it is a special case of the heterogeneous extension rather than the maintained case.
Define
$$\underline m_O
=
\inf_{\theta\in[0,1],\,j\ge1}m_\theta^O(j).$$
If the realized reward satisfies $\tau<\underline m_O$, then no organic user claims: $k_\theta^O(\tau)=0$ and $S_\theta^O(\tau)=0$ for every $\theta$.
In that special case, $a_O=0$ and $A=a_F=F_F(\tau)$.
If $F_F$ is strictly increasing and $w_F(A)=F_F^{-1}(A)$, then $\tau=w_F(A)$ and $T(A)=Aw_F(A)$.
The homogeneous model is the perfectly elastic limiting case $w_F(A)=c_F$ for all $A>0$.

With a heterogeneous upward-sloping farmer supply curve, a pure fixed reward can induce finite farmer entry: $\Omega=0$ implies $A=F_F(r)+a_O.$
This softens the fixed-reward impossibility from the homogeneous baseline.
The baseline is the perfectly elastic limiting case in which farmer marginal costs are flat: $F_F(r)=0$ for $r<c_F$ and $F_F(r)=\infty$ for $r>c_F$, with the knife-edge $r=c_F$ leaving finite activity indeterminate.

Thus heterogeneous marginal claim costs do not preserve every result of the homogeneous model.
They preserve the core procurement logic.
Under weak identity, the platform is not choosing which persons receive the airdrop.
It is buying a separable platform-building input from the lower tail of the marginal-identity cost distribution.
Some of those marginal identities may belong to organic users who farm at the margin.
If farmers dominate the lower tail, they receive most of the transfers.
The beneficiaries of the procured activity, however, need not be the same agents who receive those transfers.
Organic users can still benefit because the activity raises $v(A)$, and the platform can benefit when the monetized gain in platform quality exceeds the transfer cost.

\section{Competition and activity procurement}
\label[appendix]{sec:competition}

We now analyze a simple competitive benchmark, with the goal of clarifying the scope of the monopoly mechanism.
In the monopoly model, the platform internalizes the market-building value of procured activity.
Under rivalry, the same activity-procurement logic remains available, but the return to activity is shaped by relative-quality competition.

To keep the comparison close to the main text, users are again indexed by $\theta\in[0,1]$.
As in the monopoly model, a lower $\theta$ indicates a better match with the original platform.
We therefore interpret the original platform as the incumbent, $I$, located at~0.
The entrant, $E$, is located at~1.
We normalize the Hotelling distance parameter to one.

We focus on full market coverage: every organic user chooses one of the two platforms.
Most of the analysis focuses on the interior-share case, in which both platforms serve a positive user mass.
Boundary outcomes where one platform captures the entire market are treated separately below.

Platform $j\in\{I,E\}$ can implement rewarded activity $A_j\ge0$ using the activity-procurement logic from the main text.
Let $c$ denote the unit cost of procuring rewarded activity.
Under weak identity with farmers, $c=c_F$.
In a counterfactual in which farmers are perfectly screened out and the same activity must be procured from verified organic claimants, $c=c_O$.
Thus the same competitive procurement game can be evaluated at different activity-procurement costs.

A platform's announced policy is summarized by an access fee $p_j$ and an implemented activity level $A_j$, where $A_j=0$ corresponds to no airdrop.
This representation does not introduce a separate activity-commitment stage.
It is the reduced form of the same policy-announcement timing used in the monopoly model: platforms announce fees and reward policies before users and farmers act.

Given $(A_I,A_E,p_I,p_E)$, a user of type $\theta$ obtains utility
\[
    U_I(\theta)=v(A_I)-p_I-\theta
\]
from joining the incumbent, and
\[
    U_E(\theta)=v(A_E)-p_E-(1-\theta)
\]
from joining the entrant.
Let
\[
    \Delta\equiv v(A_E)-v(A_I).
\]
Under full coverage and interior shares, demands are
\[
    n_E=\frac12+\frac{\Delta-p_E+p_I}{2},
    \qquad
    n_I=\frac12+\frac{-\Delta-p_I+p_E}{2}.
\]
Conditional on activity levels, the unique fee best-response profile is
\[
    p_E^*=1+\frac{\Delta}{3},
    \qquad
    p_I^*=1-\frac{\Delta}{3},
\]
with market shares
\[
    n_E^*=\frac12+\frac{\Delta}{6},
    \qquad
    n_I^*=\frac12-\frac{\Delta}{6},
\]
and profits
\[
    \pi_E(A_E,A_I)
    =
    \frac12+\frac{\Delta}{3}+\frac{\Delta^2}{18}-cA_E,
\]
\[
    \pi_I(A_I,A_E)
    =
    \frac12-\frac{\Delta}{3}+\frac{\Delta^2}{18}-cA_I.
\]
These formulas are useful for interpretation, but the equilibrium analysis below keeps the policy-announcement timing: activity and fees are chosen as parts of the same announced policy.

For a platform with relative quality advantage $z$ over its rival, let the optimized full-coverage gross fee revenue be
$$R(z)
=
\begin{cases}
0,
& z\le -2,\\
\dfrac{(2+z)^2}{8},
& -2\le z\le2,\\
z,
& z\ge2.
\end{cases}$$
So $R(0)=\frac{1}{2}$.
The function $R$ maximizes over all feasible market shares, not only interior ones.
This matters because a deviating platform may choose a large airdrop that allows it to capture the entire market.

Define
$$S_0
\equiv
\sup_{A>0}
\frac{R(v(A)-v(0))-R(0)}{A}.$$
We also impose the following condition on activity levels that could appear in an interior equilibrium.

\begin{assumption}[Competitive single crossing]
\label{ass:competition_single_crossing}
For a unit activity-procurement cost $c>0$, let $\mathcal R_c$ be the set of activity levels $A>0$ at which $v$ is differentiable and $v'(A)>c$.
On $\mathcal R_c$, the function $\Psi_c(A) \equiv \frac{c}{v'(A)}- \frac{v(A)}{3}$ is one-to-one.
\end{assumption}

The restriction is imposed only on $\mathcal R_c$ as any positive interior activity choice must satisfy $$n_jv'(A_j)=c$$ with $n_j\in(0,1)$, and hence $v'(A_j)>c$.
The term $\frac{c}{v'(A)}$ is the market share needed to justify activity $A$ at unit cost $c$, and $\frac{v(A)}{3}$ is the Hotelling market-share effect of the activity's generated quality.
If $\Psi_c$ were not one-to-one, two different activity levels could support the same residual index, leaving room for asymmetric interior equilibria where one platform chooses low activity with high marginal returns while the other chooses high activity with low marginal returns.
The assumption rules out that knife-edge multiplicity; it is not used in the monopoly analysis.

\begin{proposition}[Competition as a scope result]
\label{res:competition_scope}
Consider the full-coverage reduced-form policy game described above.
\begin{enumerate}[label=(\roman*), leftmargin=2em]
\item The symmetric no-airdrop policy profile $(p_I,A_I)=(p_E,A_E)=(1,0)$ is a pure policy equilibrium at unit procurement cost $c$ if and only if $S_0\le c.$
If $v$ is differentiable at zero, a necessary local condition is $v'(0)\le2c.$

\item In any positive interior pure policy equilibrium with full coverage and positive market shares for both platforms,
\[
        p_j=2n_j,
        \qquad
        n_jv'(A_j)=c,
        \qquad
        j\in\{I,E\},
    \]
and
    \[
        n_I-n_E
        =
        \frac{v(A_I)-v(A_E)}{3}.
    \]
If \cref{ass:competition_single_crossing} holds at activity cost $c$, every such equilibrium is symmetric.
Hence every positive interior pure policy equilibrium satisfies
    \[
        A_I=A_E=A^C,
        \qquad
        n_I=n_E=\frac12,
        \qquad
        p_I=p_E=1,
    \]
and
    \[
        v'(A^C)=2c.
    \]
If $v'$ is strictly decreasing, there is at most one positive interior equilibrium activity level.
Each platform's profit at such an equilibrium is $\frac12-cA^C.$

\item Compare the farmer-available game, with unit cost $c_F$, to the screened game, with unit cost $c_O>c_F$.
If $S_0\le c_F$, then the symmetric no-airdrop profile is an equilibrium both with and without farmers.
If $c_F<S_0\le c_O$, then the symmetric no-airdrop profile is an equilibrium when farmers are screened out, but not when farmers are available.

In that intermediate case, if the farmer-available game admits a positive interior pure policy equilibrium and \cref{ass:competition_single_crossing} holds at $c_F$, then the equilibrium is symmetric with $A_F^C>0$.
Platforms are worse off than in the screened no-airdrop equilibrium, as each one earns $\frac{1}{2}-c_FA_F^C<\frac{1}{2}$.
Organic users are weakly better off, and strictly so when $v(A_F^C)>v(0)$.

If $S_0>c_O$, both environments admit positive interior pure policy equilibria, \cref{ass:competition_single_crossing} holds at both $c_F$ and $c_O$, and $v'$ is strictly decreasing, then all such equilibria are symmetric.
Denote their activity levels by $A_F^C$ and $A_O^C$.
Then $A_F^C\ge A_O^C.$
Organic users' platform-consumption payoffs are weakly higher when farmers are available.
Platform profits compare according to
$$\left(\frac12-c_FA_F^C\right)-
  \left(\frac12-c_OA_O^C\right)
  =
  c_OA_O^C-c_FA_F^C.$$
Thus farmer availability benefits platforms if and only if $c_FA_F^C\le c_OA_O^C.$
\end{enumerate}
\end{proposition}

\begin{proof}
At the symmetric no-airdrop profile, both platforms choose $A=0$, charge $p=1$, split the market, and earn $R(0)=\frac{1}{2}$.
Consider a unilateral deviation by one platform to activity $A\ge0$ and fee $p$, while the rival remains at $(1,0)$.
Let $z=v(A)-v(0)\ge0$ be the deviating platform's activity-induced quality advantage.
The deviator's market share is
$$n(p,A)
  =
  \min\left\{1,\max\left\{0,\frac{2+z-p}{2}\right\}\right\}.$$

Thus, for fixed $A$, the deviator's best gross fee revenue is $R(z)$, and its net deviation payoff is $R(v(A)-v(0))-cA.$
The no-airdrop profile is an equilibrium if and only if this payoff is no larger than $R(0)$ for every $A>0$, which is equivalent to $S_0\le c$.
For small $A$,
$$R(v(A)-v(0))-R(0)
  =
  \frac{v(A)-v(0)}{2}
  +
  \frac{[v(A)-v(0)]^2}{8},$$
so differentiability at zero gives the necessary local condition $v'(0)\le2c$.
This proves part (i).

For part (ii), consider a positive interior pure policy equilibrium with full coverage and positive market shares.
Platform $j$'s profit is $\pi_j=p_jn_j-cA_j.$
For an interior fee choice, the first-order condition with respect to $p_j$ is $$n_j-\frac{p_j}{2}=0,$$ so $p_j=2n_j.$
For an interior activity choice, the first-order condition with respect to $A_j$ is $$\frac{p_jv'(A_j)}{2}-c=0.$$
Using $p_j=2n_j$ gives $n_jv'(A_j)=c.$

Let $d\equiv n_I-n_E$.
Since $p_j=2n_j$,
\[
    p_I-p_E=2d.
\]
The demand equations imply
\[
    d=v(A_I)-v(A_E)-p_I+p_E.
\]
Substituting $p_I-p_E=2d$ gives
\[
    3d=v(A_I)-v(A_E),
\]
or
\[
    n_I-n_E=\frac{v(A_I)-v(A_E)}{3}.
\]

The activity first-order conditions imply
\[
    n_I=\frac{c}{v'(A_I)},
    \qquad
    n_E=\frac{c}{v'(A_E)}.
\]
Combining these expressions with the market-share difference gives
\[
    \frac{c}{v'(A_I)}-
    \frac{c}{v'(A_E)}
    =
    \frac{v(A_I)-v(A_E)}{3}.
\]
Rearranging,
$$
\frac{c}{v'(A_I)}-
  \frac{v(A_I)}{3}
  =
  \frac{c}{v'(A_E)}-
  \frac{v(A_E)}{3}
.
$$
That is, $\Psi_c(A_I)=\Psi_c(A_E).$
Because $n_I,n_E\in(0,1)$, the activity first-order conditions imply $A_I,A_E\in\mathcal R_c$.
\cref{ass:competition_single_crossing} therefore implies $A_I=A_E.$
Symmetry of demand then gives $n_I=n_E=\frac12,$ and the fee first-order condition gives $p_I=p_E=1.$
The activity first-order condition becomes $\frac12 v'(A^C)=c,$ or $v'(A^C)=2c.$
If $v'$ is strictly decreasing, this equation has at most one positive solution.
At a symmetric activity profile, $\Delta=0$, so each platform's gross fee revenue is $\frac{1}{2}$ and net profit is $\frac{1}{2}-cA^C$.
This proves part (ii).

For part (iii), if $S_0\le c_F$ and $c_F<c_O$, then $S_0\le c_O$, so part (i) implies that the symmetric no-airdrop profile is an equilibrium in both environments.
If $c_F<S_0\le c_O$, part (i) implies that no-airdrop is stable when farmers are screened out but not stable when farmers are available.

In the intermediate case, suppose the farmer-available game admits a positive interior pure policy equilibrium and the single-crossing condition holds at $c_F$.
By part (ii), that equilibrium is symmetric with activity $A_F^C>0$.
Platform profit is therefore $\frac{1}{2}-c_FA_F^C<\frac{1}{2}$, while the screened no-airdrop equilibrium gives each platform profit $\frac{1}{2}$.
Organic users at a symmetric activity profile $A$ have platform-consumption payoff
$$u_\theta^C(A)
  =
  \begin{cases}
  v(A)-1-\theta,
  & \theta\le\frac{1}{2},\\
  v(A)-1-(1-\theta),
  & \theta\ge\frac{1}{2}.
  \end{cases}$$
Thus $u_\theta^C(A_F^C)-u_\theta^C(0)=v(A_F^C)-v(0)\ge0,$ with strict inequality when $v(A_F^C)>v(0)$.

Finally, suppose $S_0>c_O$, both environments admit positive interior pure policy equilibria, the single-crossing condition holds at both costs, and $v'$ is strictly decreasing.
By part (ii), all such equilibria are symmetric and satisfy $v'(A_F^C)=2c_F$ and $v'(A_O^C)=2c_O.$
Since $c_F<c_O$ and $v'$ is strictly decreasing, $A_F^C\ge A_O^C.$
The organic-user comparison follows from $v(A_F^C)\ge v(A_O^C)$.
Symmetric platform profit at cost $c$ and activity $A^C$ is $\frac{1}{2}-cA^C$, so the profit comparison is exactly
$$\left(\frac12-c_FA_F^C\right)-
  \left(\frac12-c_OA_O^C\right)
  =
  c_OA_O^C-c_FA_F^C.$$
This proves the condition for farmer availability to benefit platforms.
\end{proof}

\cref{res:competition_scope} is the competitive counterpart to the main monopoly result.
In the monopoly model, lower-cost farmers can be a win-win: they lower the cost of procuring platform-building activity, raise platform profit, and weakly increase every organic user's payoff.
Under competition, the organic-user side of the comparison is similar but the platform side is not.

The reason is that a symmetric increase in activity raises both platforms' quality by the same amount.
Fee competition dissipates this common quality improvement from symmetric gross platform profits.
Organic users still benefit from the higher platform value $v(A)$, but platforms must pay the procurement cost.
Thus, lower-cost farmers can make airdrops more attractive to users while making platforms worse off by triggering a costly relative-quality contest.

The single-crossing condition is important for interpretation.
It means the symmetric positive-airdrop characterization is not merely an equilibrium selection.
Within the full-coverage, positive-share, interior-pure class, the condition rules out asymmetric activity profiles.
This characterization therefore does not hide a large set of asymmetric interior equilibria; any remaining uncharacterized cases are boundary or non-pure cases, discussed below.

This distinction is the main competitive scope result.
Farmers are beneficial to organic users whenever they lead competing platforms to procure more platform-building activity.
Farmers benefit platforms only when the lower unit procurement cost more than offsets the additional activity spending induced by competition.
In contrast to the monopoly benchmark, $c_F<c_O$ alone is not enough to guarantee that platforms benefit from farmer availability.

\subsection*{Boundary market shares}

The analysis above focuses on full-coverage equilibria with positive market shares for both platforms.
Boundary outcomes, in which one platform captures the entire market, introduce a force different from the one that this competition appendix is meant to isolate.
If one platform captures the whole market, it must set a price low enough to deter the rival from entering.
This is a limit-pricing problem.
A candidate winner trades off a higher access fee against the risk of conceding some users to the rival; a candidate loser compares staying out with a discrete activity-and-fee deviation that may let it enter or capture the market.

The no-airdrop equilibrium condition in \cref{res:competition_scope} already allows for such discrete deviations.
The function $R(\cdot)$ maximizes the deviating platform's fee revenue over all market shares, including deviations that capture the entire market.
Thus, when the proposition says that no-airdrop is an equilibrium, it rules out not only small interior deviations but also large activity deviations that would make one platform dominant.

For the positive-airdrop analysis, we focus on full-coverage pure policy equilibria with positive market shares for both platforms.
Under \cref{ass:competition_single_crossing}, such interior equilibria are necessarily symmetric, so the symmetric characterization is not sweeping aside a large set of asymmetric interior equilibria.
Boundary limit-pricing outcomes are possible in principle, but they require a separate contestability analysis and are not needed for the comparison with the monopoly model.

\section{Proofs of results in the main text}
\label[appendix]{sec:proofs}
\begin{proof}[Proof of \cref{res:fixed}]
With \(\Omega=0\), the realized reward is \(\tau(A)=r\) whenever \(A>0\), while \(\tau(0)=0\).

First suppose \(r<c_F<c_O\).
We start by showing that no organic users claim.
If \(a_O>0\), then the set \(\mathcal C\) has positive measure.
Optimality for organic claimants requires for almost every \(\theta\in\mathcal C\)
\[
    u_\theta^1(A,p)\ge u_\theta^0(A,p).
\]
But
\[
    u_\theta^1(A,p)-u_\theta^0(A,p)=\tau(A)-c_O<0,
\]
yielding a contradiction.
Hence \(a_O=0\).

Farmers also do not enter.
If \(a_F>0\), farmer entry condition in Definition~\ref{def:equilibrium} requires \( \tau(A)=c_F.\)
But if \(A>0\), then \(\tau(A)=r<c_F\), and if \(A=0\), then \(a_F>0\) is inconsistent with \(A=a_O+a_F\).
Therefore \(a_F=0\).
Since \( A=a_O+a_F, \) we obtain \(A=0\).
This proves part (i).

Now suppose \(r=c_F<c_O\).
The same argument as above implies \(a_O=0\).
Thus organic users do not claim.

Farmers, however, earn zero profit from each positive farmed identity whenever \(A>0\), because
\[
    \tau(A)=r=c_F.
\]
Hence any finite mass \(a_F>0\) is consistent with farmer finite free-entry condition.
For \(a_F=0\), we have \(A=0\), and the \(A=0\) no-entry condition is satisfied because \(\Omega=0\) and \(r=c_F\).

For any finite \(a_F\ge0\), set \( A=a_F,\,a_O=0, \) and choose \(\mathcal C=\emptyset\).
Let \(\mathcal N\) coincide up to a measure-zero cutoff with \( \{\theta\in[0,1]: v(A)-\theta-p\ge0\}.\)
Then \( n=\mu(\mathcal N) = \min\{1,\max\{0,v(A)-p\}\}, \) organic users optimize, farmers satisfy finite free-entry condition, and \(A=a_O+a_F\).
Thus every finite \(a_F\ge0\) can be supported, so activity is indeterminate.
This proves part (ii).

Finally suppose \(r>c_F\).
If \(A=0\), the farmer entry condition in Definition~\ref{def:equilibrium} requires \( \Omega=0 \quad\text{and}\quad r\le c_F, \) which is violated.
Hence no bounded continuation equilibrium can have \(A=0\).

Now consider any putative bounded continuation equilibrium with \(A>0\).
Since \(\Omega=0\), the realized reward is \( \tau(A)=r>c_F .\)
If \(a_F>0\), farmer finite free-entry requires \(\tau(A)=c_F\), a contradiction.
If \(a_F=0\), the farmer no-entry condition for \(A>0\) requires \(\tau(A)\le c_F\), again a contradiction.
Therefore no bounded continuation equilibrium exists.
This is the unbounded-entry case: every farmer-created identity earns strictly positive payoff \(r-c_F\), so the competitive farmer fringe has no finite entry level at which entry stops.
This proves part (iii), and completes the proof of \cref{res:fixed}.
\end{proof}

\bigskip

\begin{proof}[Proof of \cref{res:activity_procurement}]
Let \((n,a_O,a_F,A)\) be a bounded continuation equilibrium, and let \(\mathcal C\subseteq\mathcal N\subseteq[0,1]\) be sets satisfying Definition~\ref{def:equilibrium}.

We first show that organic users do not claim, so \(a_O=0\).
Suppose instead that \(a_O>0\).
Then \(\mathcal C\) has positive measure.
For almost every \(\theta\in\mathcal C\), organic optimality requires
\[
    u_\theta^1(A,p)\ge u_\theta^0(A,p).
\]
Since
\[
    u_\theta^1(A,p)-u_\theta^0(A,p)=\tau(A)-c_O,
\]
this implies \( \tau(A)\ge c_O.\)
Also, \(a_O>0\) implies \(A>0\).

Now consider farmer entry.
If \(a_F>0\), farmer finite free-entry requires \( \tau(A)=c_F<c_O, \) contradicting \(\tau(A)\ge c_O\).
If \(a_F=0\), then because \(A>0\), the farmer no-entry condition requires \( \tau(A)\le c_F<c_O, \) again contradicting \(\tau(A)\ge c_O\).
Hence \(a_O=0\).

Since total rewarded activity satisfies \( A=a_O+a_F, \) we have \( A=a_F.\)

Now consider the reward pool.
If \(\Omega=0\), then for every \(A>0\),
\[
    \tau(A)=r<c_F.
\]
Because \(A=a_F\), any equilibrium with \(A>0\) would have \(a_F>0\), and farmer finite free-entry would require \(\tau(A)=c_F\), a contradiction.
Therefore \(A=0\), and hence \(a_F=0\).
This agrees with
\[
    A=\frac{\Omega}{c_F-r}=0.
\]

If \(\Omega>0\), then \(A=0\) cannot be a continuation equilibrium, because the farmer no-entry condition at \(A=0\) requires \(\Omega=0\).
Hence \(A>0\).
Since \(A=a_F\), we have \(a_F>0\).
Farmer finite free-entry then requires
\[
    \tau(A)=c_F.
\]
Using the reward formula, \( r+\frac{\Omega}{A}=c_F.\)
Solving gives \( A=\frac{\Omega}{c_F-r}.\)
Because \(r<c_F\), this activity level is finite.
Thus
\[
    a_F=A=\frac{\Omega}{c_F-r}.
\]

It remains to characterize organic adoption.
If \(A=0\), then \( \tau(0)=0<c_O.\)
If \(A>0\), the previous argument gives \( \tau(A)=c_F<c_O.\)
Thus in all cases,
\[
    u_\theta^1(A,p)<u_\theta^0(A,p).
\]
Therefore organic users' relevant adoption payoff is
\[
    u_\theta^0(A,p)=v(A)-\theta-p.
\]
By organic-user optimality in Definition~\ref{def:equilibrium}, up to a measure-zero cutoff type,
\[
    \theta\in\mathcal N
    \quad\Longleftrightarrow\quad
    v(A)-\theta-p\ge0.
\]
Since \(\theta\) is uniformly distributed on \([0,1]\),
\[
    n=\mu(\mathcal N)
    =
    \min\{1,\max\{0,v(A)-p\}\}.
\]
If adoption is interior, it reduces to \( n=v(A)-p.\)
This completes the proof of \cref{res:activity_procurement}.
\end{proof}

\bigskip

\begin{proof}[Proof of \cref{res:optimal_pricing_activity}]
Fix $A$.
By \cref{res:activity_procurement}, organic adoption is $n=\min\{1,\max\{0,v(A)-p\}\}.$
For any feasible interior adoption level $n$, the inverse demand curve is $p=v(A)-n$.
Conditional fee revenue can therefore be written as $n[v(A)-n]$.
The unconstrained maximizer of this strictly concave quadratic is $\frac{v(A)}{2}$.
The adoption cap imposes $n\le1$, so $\widehat n(A)=\min\left\{\frac{v(A)}{2},1\right\}.$
Substitution gives the fee and revenue formulas.
Differentiating the two regions of $\widehat Q$ gives $\widehat Q'(A)=\widehat n(A)v'(A)$; at $v(A)=2$, the two expressions coincide whenever $v$ is differentiable.

By \cref{res:activity_procurement}, any activity level $A\ge0$ can be implemented at transfer cost $c_FA$.
Conditional on $A$, the platform chooses $p=\widehat p(A)$ and obtains fee revenue $\widehat Q(A)$.
Hence the platform's optimized payoff from activity $A$ is $\widehat Q(A)-c_FA$, and the optimal activity set is the stated argmax.

No airdrop is optimal if and only if $\widehat Q(0)\ge \widehat Q(A)-c_FA$ for every $A>0$, which is equivalent to the stated global condition.
If $A^*>0$ is an interior optimizer and $v$ is differentiable at $A^*$, the first-order condition is $\widehat Q'(A^*)-c_F=0$.
Using the derivative formula gives $\widehat n(A^*)v'(A^*)=c_F$.
This completes the proof of \cref{res:optimal_pricing_activity}.
\end{proof}

\bigskip

\begin{proof}[Proof of \cref{res:farmer_exclusion_pareto}]
Define $G_F(A)=\widehat Q(A)-c_FA$ and $G_O(A)=\widehat Q(A)-c_OA$.
For every $A\in\mathcal A$, $$G_F(A)-G_O(A)=(c_O-c_F)A\ge0.$$
Evaluating the farmer-enabled problem at $A_O^*$ gives
$$\Pi_F^*
  \ge
  G_F(A_O^*)
  =
  G_O(A_O^*)+(c_O-c_F)A_O^*
  \ge
  G_O(A_O^*)
  =
  \Pi_O^*.$$
This proves the platform-profit comparison.

Suppose toward a contradiction that $A_F^*<A_O^*$.
$A_F^*$'s farmer-enabled optimality implies
$$\widehat Q(A_F^*)-c_FA_F^*
  \ge
  \widehat Q(A_O^*)-c_FA_O^*,$$
or
$$\widehat Q(A_O^*)-\widehat Q(A_F^*)
  \le
  c_F(A_O^*-A_F^*).$$
The optimality of $A_O^*$ in the verified organic-procurement problem implies
$$\widehat Q(A_O^*)-c_OA_O^*
  \ge
  \widehat Q(A_F^*)-c_OA_F^*,$$
or
$$\widehat Q(A_O^*)-\widehat Q(A_F^*)
  \ge
  c_O(A_O^*-A_F^*).$$
Because $c_O>c_F$ and $A_O^*-A_F^*>0$, the two inequalities cannot both hold and $A_F^*\ge A_O^*$.

Next, we show that organic type payoff under the conditional optimal fee $u_\theta(A)$ is weakly increasing in $A$.
At activity level $A$, the conditional optimal fee satisfies $\widehat p(A)=v(A)-\widehat n(A)$ by \cref{res:optimal_pricing_activity}.
Hence, an organic adopter's platform-consumption payoff is $v(A)-\widehat p(A)-\theta=\widehat n(A)-\theta$.
Considering the outside option, we get $u_\theta(A)=\max\{\widehat n(A)-\theta,0\}$.
Because $v$ is increasing, $\widehat n(A)=\min\{v(A)/2,1\}$ is weakly increasing in $A$.
Therefore $A'\ge A$ implies $u_\theta(A')\ge u_\theta(A)$ for every $\theta$.
Hence, for every $\theta\in[0,1]$, $A_F^*\ge A_O^*$ implies $u_\theta(A_F^*)\ge u_\theta(A_O^*)$.
This completes the proof of \cref{res:farmer_exclusion_pareto}.
\end{proof}

\printbibliography[heading=bibintoc]

\end{document}